\documentclass[%
reprint,
superscriptaddress,
 amsmath,amssymb,
 aps,
 pra,
]{revtex4-1}

\usepackage{graphicx}
\usepackage{dcolumn}
\usepackage{bm}
\usepackage{amsthm}
\usepackage[dvipsnames]{xcolor}
\usepackage{chngcntr}
\usepackage{longtable} 
\usepackage[hypertexnames=false,colorlinks,linkcolor=blue!80!black,urlcolor=purple!50!black,citecolor=green!50!black]{hyperref}
\usepackage{algorithm}
\usepackage[noend]{algpseudocode}
\usepackage{braket}
\usepackage{dsfont}
\usepackage{booktabs}
\usepackage{multirow}

\usepackage[capitalise]{cleveref}
\crefname{figure}{Fig.}{Fig.}

\newtheorem{theorem}{Theorem}
\newtheorem*{theoremnonum}{Theorem}
\newtheorem{corollary}{Corollary}[theorem]
\newtheorem{lemma}{Lemma}
\theoremstyle{definition}

\newtheorem{criterion}{Criterion}
\theoremstyle{plain}

\theoremstyle{definition}

\theoremstyle{remark}
\newtheorem{remarknote}{Remark}

\DeclareMathOperator*\argmin{argmin}

\def\M{\mathcal{M}}
\def\maxk{{k_\text{max}}}
\def\estunitary{U_{\text{t}}}
\DeclareMathOperator\arctantwo{atan2}
\DeclareMathOperator\Binom{Binom}
\DeclareMathOperator\Norm{Norm}
\def\newterm#1{\emph{#1}}

\def\Plau{\Phi}
\def\Cons{\Lambda}

\begin{document}

\title{Consistency testing for robust phase estimation}

\author{Antonio E. Russo}
\affiliation{Center for Computing Research, Sandia National Laboratories, Albuquerque, NM 87185}

\author{William M. Kirby}
\affiliation{Department of Physics and Astronomy, Tufts University, Medford, MA 02155}

\author{Kenneth M. Rudinger}
\affiliation{Center for Computing Research, Sandia National Laboratories, Albuquerque, NM 87185}

\author{Andrew D. Baczewski}
\affiliation{Center for Computing Research, Sandia National Laboratories, Albuquerque, NM 87185}

\author{Shelby Kimmel}
\affiliation{Department of Computer Science, Middlebury College, Middlebury, VT 05753}

\begin{abstract}
We present an extension to the robust phase estimation protocol, which can identify incorrect results that would otherwise lie outside the expected statistical range.
Robust phase estimation is increasingly a method of choice for applications such as estimating the effective process parameters of noisy hardware, but its robustness is dependent on the noise satisfying certain threshold assumptions.
We provide consistency checks that can indicate when those thresholds have been violated, which can be difficult or impossible to test directly.
We test these consistency checks for several common noise models, and identify two possible checks with high accuracy in locating the point in a robust phase estimation run at which further estimates should not be trusted.
One of these checks may be chosen based on resource availability, or they can be used together in order to provide additional verification.
\end{abstract}

\pacs{Valid PACS appear here}
\maketitle

\section{Introduction}
\label{intro}

The phase estimation algorithm is ubiquitous in quantum computing.
It is common as an algorithmic primitive~\cite{kitaev1995quantum, shor1994algorithms,griffiths1996semiclassical,PhysRevX.6.031007,PhysRevLett.118.100503,da2020optimizing} and is also used for error mitigation \cite{obrien20a} and for estimating the parameters of quantum processes~\cite{dong2007general,Higgins_2009}.
However, error rates in noisy intermediate-scale quantum (NISQ) devices, particularly in state preparation and measurement (SPAM), present a challenge for implementing phase estimation in existing and near-future hardware~\cite{o2016scalable}.
This incentivizes the development of intrinsically error-resilient, or robust, protocols for phase estimation~\cite{Higgins_2009,Helsen19,o2019quantum,PhysRevLett.117.010503}.

Robust phase estimation (RPE) is one such protocol that was originally conceived as a method for characterizing single-qubit gates~\cite{Kimmel2015Dec}.
Recently, RPE implementations have been experimentally demonstrated on trapped-ion qubits~\cite{Rudinger2017May,Meier2019Nov} and used to simulate the ground state and low-lying electronic excitations of a hydrogen molecule on a superconducting cloud-based quantum computer~\cite{Russo2020Jul}.
RPE has Heisenberg scaling, so it is optimally fast up to constant factors.
It is robust to all errors below a certain threshold.
Furthermore, it is easy to implement, as it involves no entangled states, or even any additional registers beyond the register on which the gate of interest acts.

RPE is based on a non-entangled-state version of phase estimation presented in 2009 by Higgins \emph{et al.} \cite{Higgins_2009}.
The focal point of an RPE protocol is a particular unitary gate whose rotation angle is to be estimated.
The protocol involves multiple generations of experiments where this unitary of interest is repeatedly applied in longer and longer sequences.
Roughly, each generation provides an additional bit of precision to the estimate of the phase.
The protocol can tolerate a relatively high degree of inaccuracy at any given round, since future generations serve to correct the accuracy.
This tolerance is what gives the protocol its robustness to a wide range of errors \footnote{As the original proof of correctness in \cite{Kimmel2015Dec} contains a flaw (also previously noted in \cite{belliardo20a}), we reprove the robustness result in a more general form in this paper in \ \cref{app:heisenberg}.}.

The proof of robustness of RPE starts with the assumption that errors do not exceed a certain size, and then shows that under those conditions, by increasing the number of samples by a constant factor we can ensure that the estimates produced will still be accurate, and will still achieve Heisenberg scaling.
The problem with this proof is that the errors we would like RPE to be robust to are themselves often expensive to accurately characterize.
Thus it is difficult to know whether they violate the threshold required for RPE to work correctly, without resorting to costlier characterization techniques.

We address this difficulty in the present work by describing tests of the self-consistency of RPE, which can herald to the user that errors have exceeded their allowed thresholds.
In particular, for several different notions of ``consistency'' we find sets of underlying angles that can explain the RPE measurement data.
Our criteria indicate inconsistency when no such angle exists.
Additionally, using realistic error models, we numerically demonstrate that the tests do a good job of flagging when errors start causing inaccuracies in the RPE estimate.

It is important to note that in this paper we are not attempting to tightly characterize resource use in RPE, which has been the primarily focus of prior work \cite{Higgins_2009,Kimmel2015Dec,belliardo20a}.
Rather, we test whether an instance of an RPE experiment provides an estimate that is trustworthy, given there are likely aspects of the system (e.g. stochastic error processes) that may not be well-characterized, but nevertheless may impact RPE's success.
The tests we develop here are somewhat akin to statistical tests of self-consistency employed in various quantum tomography schemes \cite{blume2017demonstration,langford2013errors, wolk2019distinguishing, mogilevtsev2013cross}.
In those cases, the aim is to perform statistically rigorous testing to see if an estimate appropriately fits the data that generated it.
However, as RPE is not a tomographically complete protocol, it does not generate a fully predictive estimate for a gate set (i.e., one that can predict outcomes for any circuit using only the operations in said gate set), so we cannot simply translate the statistical tests used with tomographic protocols, and instead need to develop a different set of tools.
Nevertheless, the question our tests aim to answer is quite similar to the tomographic consistency tests: given a dataset and some parameter(s) estimated from it, ought we trust those estimates?

We begin by reviewing the RPE protocol in \cref{sec:rpe_review}.
We emphasize the multi-generational nature of RPE, which creates the opportunity for self-consistency tests.
In \cref{sec:criteria}, we define various notions of consistency that can be applied to sequences of choices of estimates across generations.
We first define four increasingly stringent tests that are related to inter-generational constraints, which we call \emph{plausible consistency}, \emph{consecutive consistency}, \emph{local consistency}, and \emph{uniform-local consistency}.
We then define three other types of consistency, which we call \emph{angular-historical consistency}, \emph{probability-historical consistency}, and \emph{intersequence consistency}, which are respectively based on an angular constraint, a statistical constraint, and consistency across different full RPE runs.
These definitions ultimately lead to a series of tests that can be applied to data from an RPE experiment in order to determine its failure point, i.e., the generation at which we should cease to trust the estimates.
We also provide a full reference implementation of the protocol in the Python package \href{https://gitlab.com/quapack/pyRPE}{\texttt{pyRPE}}.

Finally, in \cref{sec:consistency-checks} we numerically test the performance of the consistency checks defined in \cref{sec:criteria}.
We apply depolarizing, dephasing, and amplitude damping noise to simulated RPE experiments for a range of error rates and target angles.
We find that the angular-historical consistency and intersequence consistency checks outperform the other checks across the board, with both of them on average flagging failure within one generation of the actual failure point for typical rotation angles.
These two consistency checks perform similarly, so since the angular-historical consistency check requires only the original set of data, while the intersequence consistency check requires a second RPE run, we suggest using angular-historical consistency as the baseline test, and employing the intersequence consistency test as a double-check if desired.
We also find that the probability-historical consistency check flags failure before actual failure occurs with high probability, providing an option for the experimenter who wants a very conservative estimate of when failure occurs.

\section{Review of the RPE protocol}
\label{sec:rpe_review}

\begin{figure}
\graphicspath{{figures/}}
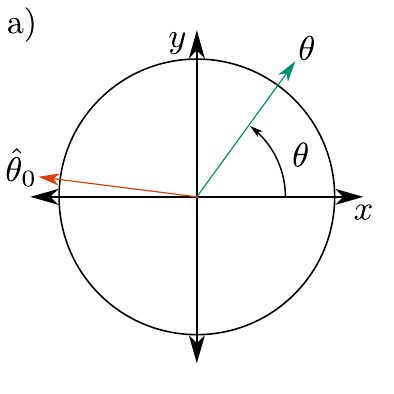
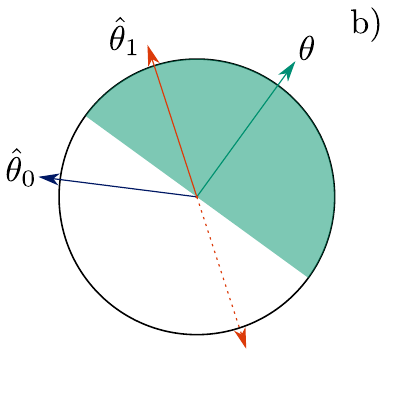\\
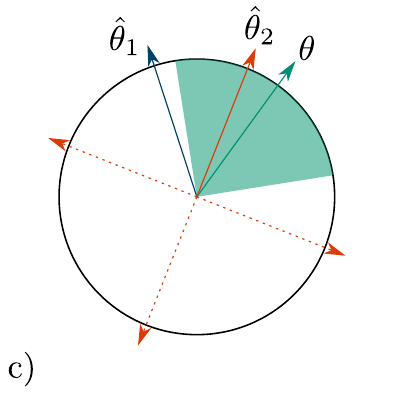
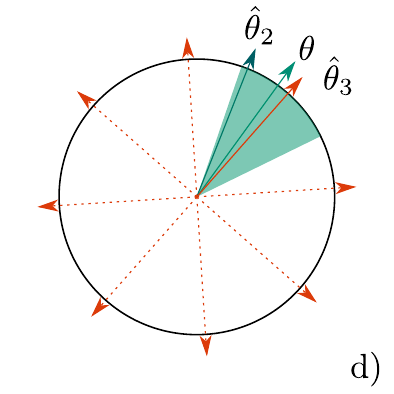
\caption{
Illustration of a successful RPE run.
Red arrows indicate elements of $\Theta_k$, and the green arrow indicates $\theta$, the correct angle to be estimated.
(a) A first measurement $\hat\theta_0$ of $\theta$ is made.
(b) A second measurement of $\theta$ is made.
Two candidate angles (red arrows) are possible, and RPE selects the one closest to the previous value, $\hat\theta_0$ (solid red).
The selected angle is within $\pi/2$ of $\theta$, as expected (within the green shaded region).
(c) and (d) As before, with four and eight candidates, and shading within $\pi/4$ and $\pi/8$ of $\theta$, respectively.
\label{fig:rpe-example-naive-success}}
\end{figure}

RPE is effectively a sequence of Ramsey and Rabi oscillation experiments with logarithmic spacing in the number of repetitions of the unitary under investigation,
\begin{equation}
    \estunitary=\exp[-i\theta\sigma_x/2],
\end{equation}
where $\sigma_x$ is the Pauli X matrix, and $\theta$ is the parameter we would like to learn.
It proceeds across multiple \newterm{generations} of experiments indexed by $k=0,1,2,...$; in the $k$th generation, $\estunitary$ is applied $N_k$ times.
At each generation, RPE produces an estimate $\hat\theta_k$ of the rotation angle $\theta$ of $\estunitary$, by combining data from prior generations.
$N_k$ is chosen such that it increases with each generation, which refines the estimate, as we will see.
In what follows we consider the requirements for implementing RPE for a single-qubit gate, noting that generalization to multi-qubit unitaries is relatively straightforward~\cite{Russo2020Jul}.

The RPE protocol requires the ability to (i) apply $\estunitary$ repeatedly,
and (ii) prepare the states $\ket{0}$ and
\begin{equation}
    \ket{i}\equiv\frac{1}{\sqrt{2}}(\ket{0}+i\ket{1}).
\end{equation}
Using these, we can construct circuits for which the distribution of outcomes encodes $\theta$:
\begin{align}
    P_{\mathrm{c},k} &= |\bra{0} \estunitary^{N_k} \ket{0}|^2
    = \frac{1}{2}\left(1+\cos(N_k \theta)\right), \label{eq:cosine_distribution} \\
    P_{\mathrm{s},k} &= |\bra{i} \estunitary^{N_k} \ket{0}|^2
    = \frac{1}{2}\left(1+\sin(N_k \theta)\right).\label{eq:sine_distribution}
\end{align}

In generation $k$ the circuits represented by \cref{eq:cosine_distribution} and \cref{eq:sine_distribution} are sampled sufficiently many times to generate estimates $\hat{P}_{\mathrm{c},k}$ and $\hat{P}_{\mathrm{s},k}$ of $P_{\mathrm{c},k}$ and $P_{\mathrm{s},k}$, respectively, from the relative frequencies of $0$ and $1$ measurement outcomes.
We do not specify the number of samples that should be taken for each circuit in the protocol, since this has been addressed in previous works \cite{Higgins_2009,belliardo20a}, and our consistency tests are agnostic to the sampling schedule.
The estimates $\hat{P}_{\mathrm{c},k}$ and $\hat{P}_{\mathrm{s},k}$ may be reinserted into \cref{eq:cosine_distribution} and \cref{eq:sine_distribution} to give us a set of candidate estimates of $\theta$ that are compatible with the experimental data,
\begin{equation}
    \Theta_k = \left\{ \tilde\theta\in[0,2\pi)~\middle|~\left(\begin{array}{c}\cos(N_k\tilde\theta)\\ \sin(N_k\tilde\theta)\end{array}\right) = \left(\begin{array}{c}2\hat{P}_{{\mathrm{c},k}}-1\\ 2\hat{P}_{{\mathrm{s},k}}-1\end{array}\right) \right\}.\label{eq:candidate_estimates}
\end{equation}
Henceforth, all angles are implicitly assumed to be defined modulo $2\pi$.
Eq.~\ref{eq:candidate_estimates} can be rewritten as
\begin{multline}
    \Theta_k = \bigg\{ \tilde\theta~\bigg|~\exists n\in\mathbb{Z}~:\\\tilde\theta = \arctantwo\left(2\hat P_{\mathrm{s},k}-1,2\hat P_{\mathrm{c},k}-1\right)/N_k + \frac{2\pi n}{N_k} \bigg\},
    \label{eq:Theta_candidates}
\end{multline}
where $n$ indexes the choice of branch in the branch cut for the arctangent.
If $N_k=1$, this set contains a unique estimate for $\theta$.
More generally for $N_k\geq 1$, there are $N_k$ candidate estimates, one in each angular interval $\left((2n-1)\frac{\pi}{N_k},(2n+1)\frac{\pi}{N_k}\right]$, for $0 \leq n< N_k$.
See the dashed red lines in \cref{fig:rpe-example-naive-success} for a graphical illustration of the angles in this set.

We want to select a single estimate, $\hat{\theta}_k$, from $\Theta_{k}$ at each generation.
The criterion used by RPE is to successively choose the $\hat\theta_k$ that is closest to the previous estimate $\hat\theta_{k-1}$.
If we assume that the initial estimate is unique, i.e., that $N_0=1$, this is possible with probability $1$\footnote{If $\hat\theta_{k-1}$ falls exactly half-way between two members of $\Theta_{k}$, we are faced with an ambiguity.
Because this occurs on a set of measure zero and adds distracting complexity to the discussion, a discussion of this contingency is deferred to \cref{app:derive-criteria}.}.
To determine which angle is closest to the previous estimate we need a branch cut-independent metric for measuring the distance between angles,
\begin{equation}
\label{eq:distance_def_pt}
|\theta'-\theta''|_{2\pi} = \min\left\{\left|\theta'-\theta''+2\pi n\right|~\big| ~n\in\mathbb{Z}\right\}.
\end{equation}
We extend this metric to define the distance from any single angle $\theta'$ to a set of angles $\Theta$,
\begin{equation}
\label{eq:distance_def_set}
    d(\theta',\Theta)=\min_{\tilde\theta\in\Theta}\left|\tilde\theta-\theta'\right|_{2\pi},
\end{equation}
the minimizer of which is
\begin{equation}
\label{eq:minimizer_def}
    \M(\theta',\Theta) = \argmin_{\tilde\theta\in\Theta}\left|\tilde\theta-\theta'\right|_{2\pi}.
\end{equation}
Using these definitions, \cref{alg:rpe-basic} states the RPE protocol.

While the above description of RPE is made in terms of the selection of the closest angle to the previous selection as in \cite{Rudinger2017May}, one can consider a procedure where a set of estimate angles is updated at each generation.
In \cref{app:interval-form}, we address this formulation of RPE and find that it has equivalent error tolerance to the single-angle approach.

Also, as mentioned above, our goal is not to calculate the resources required from first principles, but to determine if the resources actually used in an experiment are sufficient, given uncharacterized noise in the system.
However, any RPE-like protocol can still achieve Heisenberg scaling in the presence of bounded noise when the number of samples increases by a constant factor, as we show in \cref{app:heisenberg} \footnote{A similar analysis of the robustness of RPE can be found in \cite{Kimmel2015Dec}, but as mentioned previously, there is an error in the details of that analysis, so we reprove the result here.}.
Thus for optimal efficiency, we suggest starting with a number of samples as prescribed by Refs.~\cite{Higgins_2009,belliardo20a}, which perform a detailed analysis in the error-free case, and then scaling the number of samples by an amount that you believe will overcome your errors according to the scaling of \cref{app:heisenberg}.
Then use the consistency tests we lay out in the next section to check if the number of samples has been increased sufficiently to overcome the actual errors, or to test at what point in the protocol the noise becomes too large to compensate for.

We suspect that many experimentalists will not actually use the complex sampling schedules suggested in Refs.~\cite{Higgins_2009,belliardo20a}, but rather will take a constant number of samples at each generation, as this schedule achieves near optimal resource efficiency, while being much simpler.
In fact, this is what we do in our own numerical simulations.
Both the consistency tests presented here, as well as the analysis in \cref{app:heisenberg}, can be applied to any sampling schedule.

\begin{figure}\begin{algorithm}[H] 
\caption{Robust Phase Estimation}
\label{alg:rpe-basic}
\begin{algorithmic}[1]
    \item[\textbf{Input:}]
    \\$\{ \Theta_{k'} \}_{k'\leq \maxk}$, the list of candidate estimates for each generation, \cref{eq:Theta_candidates}
    \item[\textbf{Output:}]
    \\$\hat\theta_\maxk$, the estimate for the underlying angle
    \item[\textbf{Preconditions:}]
    \\$N_{k'} < N_{k}$ for $k'<k$
    \\$N_0=1$
    \item[\textbf{Code:}]
\Function{RobustPhaseEstimation}{$\{\Theta_{k'}\}_{k'\leq\maxk}$}
    \State $k\leftarrow 0$
    \State $\hat \theta_0 \leftarrow$ the unique element in $\Theta_0$.
    \While{$k+1\leq\maxk$}
        \State $k\leftarrow k+1$
        \State $\hat\theta_k\leftarrow \M(\hat\theta_{k-1},\Theta_k)$
    \EndWhile
    \State \algorithmicreturn{} $\hat\theta_{\maxk}$
\EndFunction
\end{algorithmic}
\end{algorithm}\end{figure}

\section{The Success and Failure of RPE}
\label{sec:criteria}

Ideally, the value of $\hat{\theta}_k$ closest to $\theta$ would be chosen at each generation.
Then
$\hat\theta_k$ would  estimate $\theta$ with error at most $\pi/N_k$ because the $N_k$ elements of $\Theta_k$ are equally spaced around the unit circle:
\begin{equation}
\label{failure_condition}
    \hat{\theta}_k=\M(\theta,\Theta_k)\quad\Leftrightarrow
    \quad|\hat{\theta}_k-\theta|_{2\pi}<\frac{\pi}{N_k}
\end{equation}
(see \cref{fig:rpe-example-naive-success} for an illustration).
If, for any reason, the RPE procedure selects a value of $\hat\theta_k$ that does not satisfy \eqref{failure_condition} at some generation $k$, we say that the procedure failed at generation $k$.

Generally, we would like to know whether RPE has failed at any given generation.
However, to know this with certainty would require knowledge of $\theta$ (required to evaluate \eqref{failure_condition}), which is the parameter being sought.
Instead, we develop heuristic consistency checks that evaluate some related conditions but are experimentally accessible, in order to approximate the maximum value of $k$ for which \eqref{failure_condition} holds.

Note that in order for RPE to succeed at some generation, it is not necessary in general for it to have succeeded at all previous generations.
This can happen when RPE fails due to a sufficiently large error in one generation, which is corrected by another sufficiently large error in a subsequent generation.
An example illustrating this is shown in \cref{fig:rpe-example-flip-success}.
However, such a success mode is not trustworthy, since it depends on the confluence of two errors, each of which would on its own be sufficient to induce failure.

\begin{figure}
\graphicspath{{figures/}}
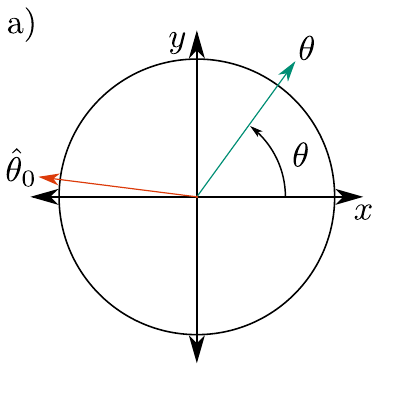
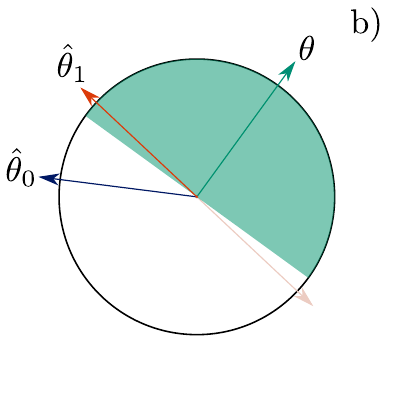\\
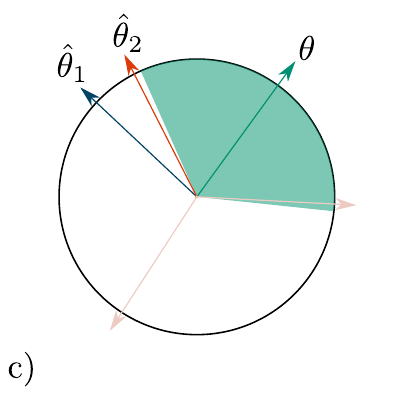
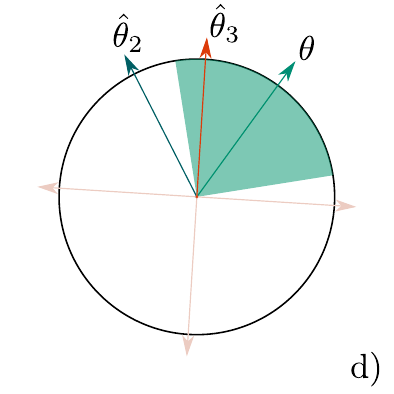
\caption{
Illustration of the action of the RPE protocol that fails at intermediate generations, but succeeds at the final generation.
Red arrows are elements of $\Theta_k$, and the green arrow is the correct angle be estimated, $\theta$.
(a) and (b) As in \cref{fig:rpe-example-naive-success}.
(c) An incorrect choice of $\hat \theta_2$ occurs---the candidate that is closest to $\theta$ was not chosen.
(d) Despite this, at generation $3$, the correct choice of $\hat\theta_3$ is still made.
\label{fig:rpe-example-flip-success}}
\end{figure}

All of the criteria for our checks are based on different notions of \emph{consistency}.
These notions are all based on properties that are satisfied in an ideal RPE run, but might fail to be realized in an RPE run in the presence of noise.
The first five criteria are based on increasingly stringent constraints on the inter-generational consistency of $\hat{\theta}_{k}$.
We define another criterion based on a condition on the inter-generational probability estimates.
Finally, we consider a criterion based on intersequence consistency across different full RPE runs.
These definitions provide a series of tests that can be applied to data from an RPE experiment in order to approximately determine its failure point, i.e., the generation at which \eqref{failure_condition} cease to be satisfied, and we should no longer trust the estimates.

Our first criterion is that there is some ``plausible angle'' $\tilde\theta$ that satisfies \cref{failure_condition} at \emph{every} generation, ruling out the situation in \cref{fig:rpe-example-flip-success}, as well as more typical failure modes like large drift in the estimates.

\begin{criterion}[plausible consistency]
\label{plausible_crit}
Consider the set of angles
\begin{align}
\label{X_k_def}
\Plau_k&=\left\{\tilde{\theta}~\middle|~\hat{\theta}_k=\M(\tilde{\theta},\Theta_k)\right\}\\
&=\left\{\tilde{\theta}~\middle|~|\hat{\theta}_k-\tilde\theta|_{2\pi}=d(\tilde{\theta},\Theta_k)\right\}\\
&=\left\{\tilde\theta~\middle|~|\hat\theta_k-\tilde\theta|_{2\pi}<\pi/N_k \right\},
\end{align}
i.e., angles for which the choice of $\hat\theta_k$ would \emph{not} constitute a failure based on \cref{failure_condition} at generation $k$.
We then check whether such an angle exists in common to all generations, giving us our first criterion:
\begin{equation}
\label{plausible_consistency}
    \bigcap_{k=0}^\maxk \Plau_k \neq \emptyset.
\end{equation}
Notice that there is no dependence on $\theta$ for this criterion, and $\lbrace \Plau_{k} \rbrace_{k\leq \maxk}$ are experimentally derivable quantities.
\end{criterion}

Because $\Plau_k$ are intervals (with size less than $\pi$ for $k>0$), their intersection is an interval, which permits efficient classical testing of membership in the intersection.
For the later criteria, care will have to be taken to ensure that criterion satisfaction is efficiently testable.

\begin{figure}
\graphicspath{{figures/}}
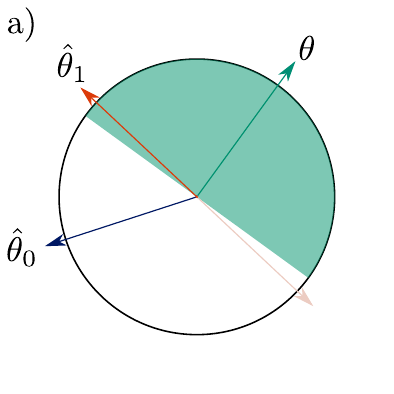
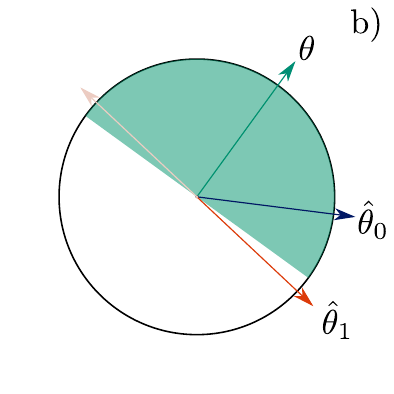
\caption{
Pathology in $\Plau_k$: increasing the accuracy of $\theta_{k-1}$ can make $\theta_k$ \emph{incorrect}.
(a) The correct $\hat\theta_1$ is chosen, despite the large error in $\hat\theta_0$.
(b) Even though $\hat\theta_0$ is closer to the true angle $\theta$ than in (a), it is closer to the incorrect candidate for generation $1$, leading to the incorrect choice of $\hat\theta_1$.
In (a), paradoxically, $\theta_0$ laying further from $\theta$, caused the correct selection of $\theta_1$ for generation $1$.
\label{fig:pathology_Xk}}
\end{figure}

Unfortunately, satisfying \cref{plausible_consistency} does not provide particularly strong guarantees.
The condition is trivially satisfied if $N_k\geq 2N_{k-1}$ for all $k$ (see \cref{app:remark_X_k_subset} in \cref{app:derive-criteria}).
Also, the criterion fails to rule out the paradoxical scenario of \cref{fig:pathology_Xk}, in which \emph{more reliable} data can lead to an incorrect choice of angle.
Our next criterion addresses these two concerns by relying on the distance to the \emph{set} $\Theta_k$, rather than the \emph{point} $\hat\theta_k$.

\begin{criterion}[consecutive consistency]
\label{consec_crit}
Consider the sets $\Cons_k$, where $\Cons_0=[0,2\pi)$ and
\begin{equation}
    \Cons_k=\left\{\tilde{\theta}~\middle|~d(\tilde\theta,\Theta_k)+d(\tilde\theta,\Theta_{k-1})<\frac{\pi}{N_k}\right\}
\end{equation}
for $k=1,2,...$.
An angle $\tilde \theta$ is in $\Cons_k$ if there exist measurements in both $\Theta_k$ and $\Theta_{k-1}$ that are close to $\tilde \theta$.
As before, the corresponding consistency test is whether the intersection of the sets $\Cons_k$ is nonempty:
\begin{equation}
\label{eq:cons_crit}
    \bigcap_{k=0}^\maxk \Cons_k \neq \emptyset.
\end{equation}
\end{criterion}

\cref{consec_crit} is stronger than \cref{plausible_crit}:
\begin{theorem}
\label{well_approx_thm}
The sets $\lbrace \Plau_{k} \rbrace_{k\leq \maxk}$ and $\lbrace \Cons_{k} \rbrace_{k\leq \maxk}$ satisfy:
\begin{equation}
\bigcap_{k'\leq k} \Cons_{k'} \subseteq \bigcap_{k'\leq k} \Plau_{k'}.
\end{equation}
Hence, \cref{eq:cons_crit} implies \cref{plausible_consistency}.
\end{theorem}
\noindent
The proof may be found in \cref{app:derive-criteria}.
Moreover, for a sequence that satisfies \cref{consec_crit} and \cref{failure_condition} for all $k'<k$,
further reduction of $d(\theta,\Theta_{k'})$ will only improve the estimate $\hat\theta_k$ (see \cref{app_cor:well_approx-stable} in \cref{app:derive-criteria}).

Unlike $\Plau_{k}$, $\Cons_{k}$ is \emph{not} an interval (for $k>0$), and testing for membership could introduce exponential classical overhead (since $\Cons_{k}$ is the union of $N_k$ intervals).
Fortunately, $\Cons_{k}\cap \Plau_k$, and hence $\bigcap_{k'\leq k}\Cons_{k'}$, \emph{is} an interval:
\begin{equation}
\Cons_{k}\cap\Plau_k = \left\{ \theta ~\middle|~ d\left(\theta,(\hat\theta_k,\hat\theta_{k-1})\right)< D_k\right\},
\end{equation}
where $D_{k}=\frac{\pi}{2N_{k}}-\frac{1}{2}\left|\hat\theta_{k}-\hat\theta_{k-1}\right|_{2\pi}$.
The interval is, equivalently, the smallest angular interval containing $\hat\theta_{k}$ and $\hat\theta_{k-1}$, expanded by $D_{k}$ on both sides.
Refer to \cref{app:consecutive_consistent_interval} in \cref{app:derive-criteria} and the \href{https://gitlab.com/quapack/pyRPE}{reference Python implementation} for further details.

\begin{table*}
\begin{tabular}{m{0.3cm}m{2cm}m{6cm}m{3.1cm}m{0.1cm}m{5.452cm}}
\multicolumn{6}{c}{Criteria of fundamental value}\\
\toprule[0.9pt]
\# & Description & Set & Criterion && Interval \\[-0.8pt]\midrule[0.1pt]\\[-7pt]

1& plausible & $\Plau_k= \left\{\tilde\theta~\middle|~\hat{\theta}_k=\M(\tilde{\theta},\Theta_k)\right\}$
&$\bigcap_{k\leq \maxk} \Plau_k\neq \emptyset$
&&$\bigcap_{k'\leq k} \left\{\tilde{\theta}~\middle|~|\tilde\theta-\hat\theta_{k'}|_{2\pi}<\frac{\pi}{N_{k'}}\right\}$\\\\[-1.75pt]

2& consecutive & $\Cons_k= \left\{\tilde\theta~\middle|~d(\tilde\theta,\Theta_k)+d(\tilde\theta,\Theta_{k-1})<\frac{\pi}{N_k}\right\}$
&$\bigcap_{k\leq \maxk} \Cons_k\neq \emptyset$
&&$\bigcap_{k'\leq k} \left\{ \theta ~\middle|~ d\left(\theta,(\hat\theta_k,\hat\theta_{k-1})\right)< D_k\right\}$ \\
&&&&&\multicolumn{1}{r}{$D_{k'}=\frac{\pi}{2N_{k'}}-\frac{1}{2}|\hat\theta_{k'}-\hat\theta_{k'-1}|_{2\pi}$}\\\\[-1.75pt]

3& local & $\Delta_k[\delta\theta_k]= \left\{\tilde\theta~\middle|~d(\tilde \theta,\Theta_k) < \frac{\delta\theta_k}{N_k}\right\}$
&$\bigcap_{k\leq \maxk} \Delta_k[\delta\theta_k]\neq \emptyset$
&&$\bigcap_{k'\leq k} \left\{\tilde{\theta}~\middle|~|\tilde\theta-\hat\theta_{k'}|_{2\pi}<\frac{\delta\theta_{k'}}{N_{k'}}\right\}$\\
&&\multicolumn{1}{l}{s.t. $\frac{\delta\theta_k}{N_k} +\frac{\delta\theta_{k-1}}{N_{k-1}} \leq \frac{\pi}{N_k}$}\\
\bottomrule[0.9pt]\\
\end{tabular}
\def\extratabspace{\vspace{0pt}}
\begin{tabular}{m{0.3cm}m{2cm}m{4.5cm}m{0.1cm}m{10.2cm}}
\multicolumn{5}{c}{Criteria useful as tests}\\
\toprule[0.9pt]
\# & Description & \multicolumn{3}{c}{Criterion} \\[-0.8pt]\midrule[0.1pt]\\[-7pt]

4& uniform local & & &$\bigcap_{k\leq \maxk} \Delta_k[\pi(1 +N_k/N_{k-1})]\neq\emptyset$ \\\\[-1.75pt]

5& \parbox[c]{2cm}{\raggedright angular-historical} & \multicolumn{1}{r}{$\forall k\leq \maxk\leq \maxk$} && $\hat\theta_{k} \in \bigcap_{k'\leq k} \Delta_{k'}[\delta\theta_{k'}]$\\\\[-1.75pt]

6& \parbox[c]{2cm}{\raggedright probability-historical} & \multicolumn{1}{r}{$\forall k'<k\leq\maxk$} &&$\max\left\{|\sin N_{k'}\hat\theta_{k'}-\sin N_{k}\hat\theta_k|,|\cos N_{k'}\hat\theta_{k'}-\cos N_{k}\hat\theta_k|\right\} \leq \frac{\sin(\delta\theta_k)}{\sqrt{2}}$\\\\[-1.75pt]

7& intersequence  & \multicolumn{1}{r}{$\forall k<\maxk$}&& $|\hat\theta_k-\hat\theta'_k|_{2\pi} \leq \frac{2\pi}{N_k}$\\
\bottomrule[0.9pt]
\end{tabular}

\caption{
Overview of consistency checks of RPE success.
The criteria in the first table (of ``fundamental value'') are used to characterize basic convergence properties of RPE, and include Criteria 1, 2, and (to a lesser extent) 3.
Notice they form a logical hierarchy, with later criteria strictly stronger than earlier ones.
The second table lists the remainder of the criteria discussed, which we find have value for empirically testing the validity of RPE results.
There is no straightforward logical hierarchy among these final 4 criteria.
\label{table:tests}
}
\end{table*}

\cref{consec_crit} is based on balancing errors on adjacent generations.
While adaptive protocols, such as \cite{Grinko2019Dec}, include detection and adjustment for errors between generations, standard RPE is non-adaptive and makes no such adjustments.
Thus a criterion that allows a small error at one generation to compensate for a larger error on another, independent of the pre-determined number of samples to take at each generation, should raise suspicion.
This motivates a third, still stronger criterion which forbids balancing errors across subsequent generations.

\begin{criterion}[local consistency]
\label{crit:delta_theta-consistency}
Consider a set of angular error bounds on the sequence, $\lbrace \delta\theta_{k} \rbrace_{k\leq \maxk}$, which should satisfy \cref{crit:delta_theta-approx} but may otherwise may be freely chosen.
At generation $k$, the set of angles within those bounds is
\begin{equation}\label{eq:delta_theta_set}
\Delta_k[\delta\theta_{k}] = \left\{\tilde \theta~\middle|~d(\tilde \theta,\Theta_k) < \frac{\delta\theta_k}{N_k}\right\}.
\end{equation}
We say that a sequence is ($\delta\theta_{k}$)-locally-consistent if the intersection across all generations is non-empty:
\begin{equation}\label{eq:delta_theta_set_int}
    \bigcap_{k\leq \maxk} \Delta_{k}[\delta\theta_k] \neq \emptyset.
\end{equation}
\end{criterion}

Local consistency is stronger than \cref{consec_crit} if
\begin{equation}\label{crit:delta_theta-approx}
\frac{\delta\theta_k}{N_k} +\frac{\delta\theta_{k-1}}{N_{k-1}} \leq \frac{\pi}{N_k},
\end{equation}
since in that case we have
\begin{equation}\label{eq:cond_sequence}
\bigcap_{k'\leq k} \Delta_{k'}[\delta\theta_{k'}]\subseteq \bigcap_{k'\leq k} \Cons_{k'} \subseteq \bigcap_{k'\leq k} \Plau_{k'}.
\end{equation}
Like $\Cons_k$, $\Delta_k$ is not necessarily an interval, but it is if the set inclusion \cref{eq:cond_sequence} holds (see \cref{app:plau_def2}), so we demand that \cref{crit:delta_theta-approx} be satisfied.
In this case, an interval formulation of this criterion follows directly from \cref{eq:delta_theta_set}:
\begin{equation}
\bigcap_{k'\leq k} \Delta_{k'}[\delta\theta_{k'}] = \bigcap_{k'\leq k} \left\{\tilde{\theta}~\middle|~|\tilde\theta-\hat\theta_{k'}|_{2\pi}<\frac{\delta\theta_{k'}}{N_{k'}}\right\}
\label{criterion:uniform_consistency-interval}
\end{equation}

\begin{criterion}[uniform-local consistency]
\label{crit:unif-approx-consitency}
From \cref{crit:delta_theta-consistency}, we can obtain a natural special case in which the $\delta\theta_k$ have a fixed dependence on the length of the sequences between subsequent generations, while satisfying  \cref{crit:delta_theta-approx}:
\begin{equation}
\label{delta_theta_k}
\delta\theta_k = \frac{\pi}{1 +\frac{N_k}{N_{k-1}}}.
\end{equation}
\end{criterion}
This balances the error tolerance between generations, so it is a natural choice when the expected error on each generation is the same, which could occur for example when SPAM error is independent of generation and dominates error associated with implementing $U_\mathrm{t}^{N_k}$.
Note that for the standard RPE case in which $N_k=2^k$ \cite{Kimmel2015Dec},
\begin{equation}
\delta\theta_k = \frac{\pi}{3}.
\end{equation}

\cref{crit:delta_theta-consistency} (including the special case \cref{crit:unif-approx-consitency}) leads to a bound on the errors in the estimates of the probabilities of \cref{eq:cosine_distribution} and \cref{eq:sine_distribution}.
A straightforward geometrical argument~--- found in, for example, \cite{Higgins_2009} (the ``simple geometry'' leading to Eq.~2) or \cite{berg2019practical} (Theorem 4.1)~--- shows that
\begin{equation}
\max\{|\hat P_{\mathrm{c},k}-P_{\mathrm{c},k}|,|\hat P_{\mathrm{s},k}-P_{\mathrm{s},k}|\}=\Delta P \leq\frac{\sin(\delta\theta_k)}{2\sqrt{2}}.
\label{eq:criterion-prob_bound}
\end{equation}
In particular, for the uniform case and $N_k=2^k$, we obtain the bound
\begin{equation}
\Delta P_k \leq \sqrt{\frac{3}{32}} \sim 30.6\%.
\end{equation}
which corrects the numerical value given in \cite{Kimmel2015Dec}.

The previous tests are based on the existence of an angle $\tilde\theta$ that is consistent with $\{\Theta_k\}_{k\leq \maxk}$, but do not indicate whether the final output $\hat\theta_k$ is such a witness.
The following criterion asserts precisely this.
\begin{criterion}[angular-historical consistency]
\label{unif_approx_pi_over_3}
For all $k\leq \maxk$,
\begin{equation}
\label{angular_interval_consistency_eq}
\hat\theta_k \in \bigcap_{k'\leq k} \Delta_{k'}[\delta\theta_{k'}].
\end{equation}
\end{criterion}
If this condition holds, the value of $\hat\theta_\maxk$ is not just the terminating measurement in a reasonable RPE measurement sequence, it is also one of the putative underlying angles.
In \cref{thm:ang_hist_equiv}, we show that, if $\delta\theta_{k'-1}/N_{k'-1} > \delta\theta_{k'}/N_{k'}$ for all $k'\leq k$ (which holds for $N_k=2^k$), \eqref{angular_interval_consistency_eq} is equivalent to requiring that the intersection has length greater than $\delta\theta_k/N_k$ for all $k\leq\maxk$:
\begin{equation}
\label{angular_interval_consistency_eq_alt}
\left|\bigcap_{k'\leq k} \Delta_{k'}[\delta\theta_{k'}]\right|>\frac{L}{N_k},
\end{equation}
where $L=\delta\theta_k$.
One might imagine fine-tuning the interval length, $L$, to optimize the performance of the consistency test in identifying the actual failure point.
In \cref{sec:consistency-checks}, we provide numerical evidence that although such fine-tuning may be possible, it is too sensitive to the error model and the value of the actual angle $\theta$ to provide a consistent advantage over using $L=\delta\theta_k$.

The previous five criteria form a hierarchy of consistency checks that are increasingly stringent.
We also consider two other criteria that do not strictly fit into this hierarchy.
The first is to directly test the probability condition \cref{eq:criterion-prob_bound} for each estimate:
\begin{criterion}[probability-historical consistency]
For all $k'<k\leq\maxk$,
\begin{equation}
|\sin N_{k'}\hat\theta_{k'}-\sin N_{k}\hat\theta_k| \leq \frac{\sin(\delta\theta_k)}{\sqrt{2}},
\end{equation}
and
\begin{equation}
|\cos N_{k'}\hat\theta_{k'}-\cos N_{k}\hat\theta_k| \leq \frac{\sin(\delta\theta_k)}{\sqrt{2}}.
\end{equation}
\end{criterion}
In other words, probability-historical consistency tests that the probabilities expected from each estimated phase (from \cref{eq:cosine_distribution} and \cref{eq:sine_distribution}) are consistent with the probabilities obtained in previous generations, up to the $\delta\theta_k$ bounds as given in \cref{eq:criterion-prob_bound}.
Because this test is expressed in terms of probabilities rather than angles, it will turn out to be overly pessimistic in the presence of incoherent noise (see \cref{sec:consistency-checks}), but it does provide a conservative estimate of the failure point.

Our final test is to compare the results of RPE originating from different sequences of $N_k$:
\begin{criterion}[intersequence consistency]
For two sequences of RPE estimates $\hat\theta_k$ and $\hat\theta_k'$ with sequences $N_k<N'_k$, check that
\begin{equation}
|\hat\theta_k-\hat\theta_k'|_{2\pi} \leq \frac{2\pi}{N_k}
\label{interseq_condition}
\end{equation}
for all $k\leq \maxk$.
\end{criterion}
Instead of looking at the single original sequence $N_k$, and checking for self-consistency, we consider a second sequence $N'_k$, and check that the resulting estimates are consistent with those of the original sequence.
Notice first that if \cref{interseq_condition} fails for some generation $k$, then either
\begin{equation}
    |\hat\theta_k-\theta|_{2\pi}>\frac{\pi}{N_k}
\end{equation}
or
\begin{equation}
    |\hat\theta'_k-\theta|_{2\pi}>\frac{\pi}{N_k}>\frac{\pi}{N'_k},
\end{equation}
i.e., at least one of the sequences does not satisfy the true condition for correctness, \cref{failure_condition}.
Using the notation $\Plau_k$ (and $\Plau'_k$, respectively) of \cref{X_k_def}, the intersequence consistency condition \cref{interseq_condition} is equivalent to
\begin{equation}
    \Plau_k\cap \Plau'_k\neq\emptyset.
\end{equation}
In other words, the intersequence consistency check tests whether there exist any plausible estimates that are consistent with both sequences.
In \cref{sec:consistency-checks} we find that this approach can provide a very good test of data, but it does of course require an extra sequence's worth of additional experimental data.

We thus have a range of consistency checks that we can use to gain information about the performance of an RPE run.
The consistency checks are summarized in \cref{table:tests}.
In the next section, we test and compare the performance of these consistency checks by simulating noisy RPE runs.

\section{Numerical Performance of Self-consistent Criteria}
\label{sec:consistency-checks}

\begin{figure*}
\centering
\includegraphics{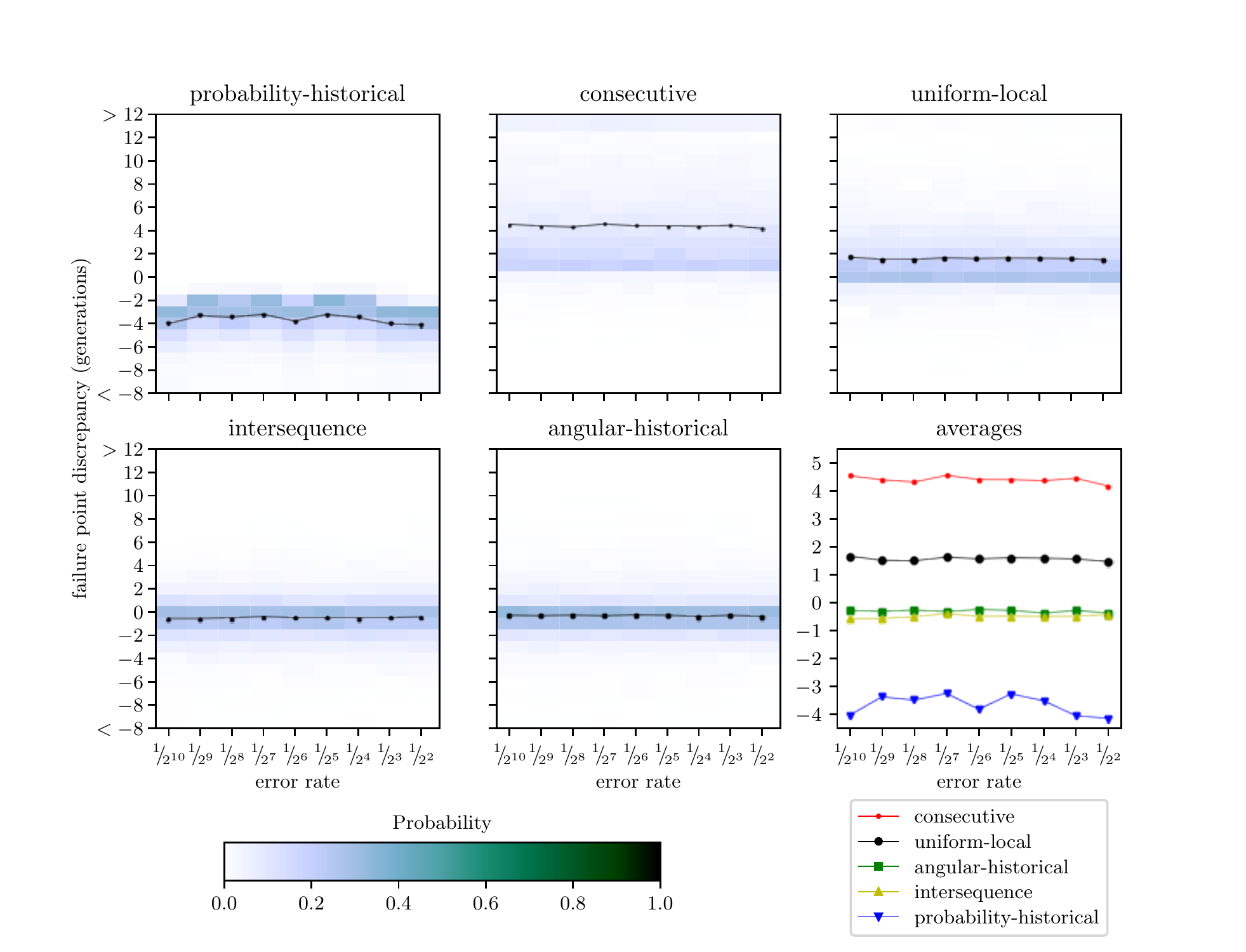}
\caption{
For varying depolarization error rates and actual angle $\theta=1.6$, plots the discrepancies between the generations at which each consistency check flagged failure, and the generations at which failure actually occurred.
Positive values indicate that the consistency check flagged failure \emph{after} failure actually occurred.
1000 RPE runs were performed per error rate.
Bin colors show the proportion of the runs at that error rate they contain.
The curves show averages for each error rate.
The lower right plot shows the average curves from all of the other plots, for comparison.
\label{depol_error_discrepancies}
}
\end{figure*}

\begin{figure*}
\centering
\includegraphics{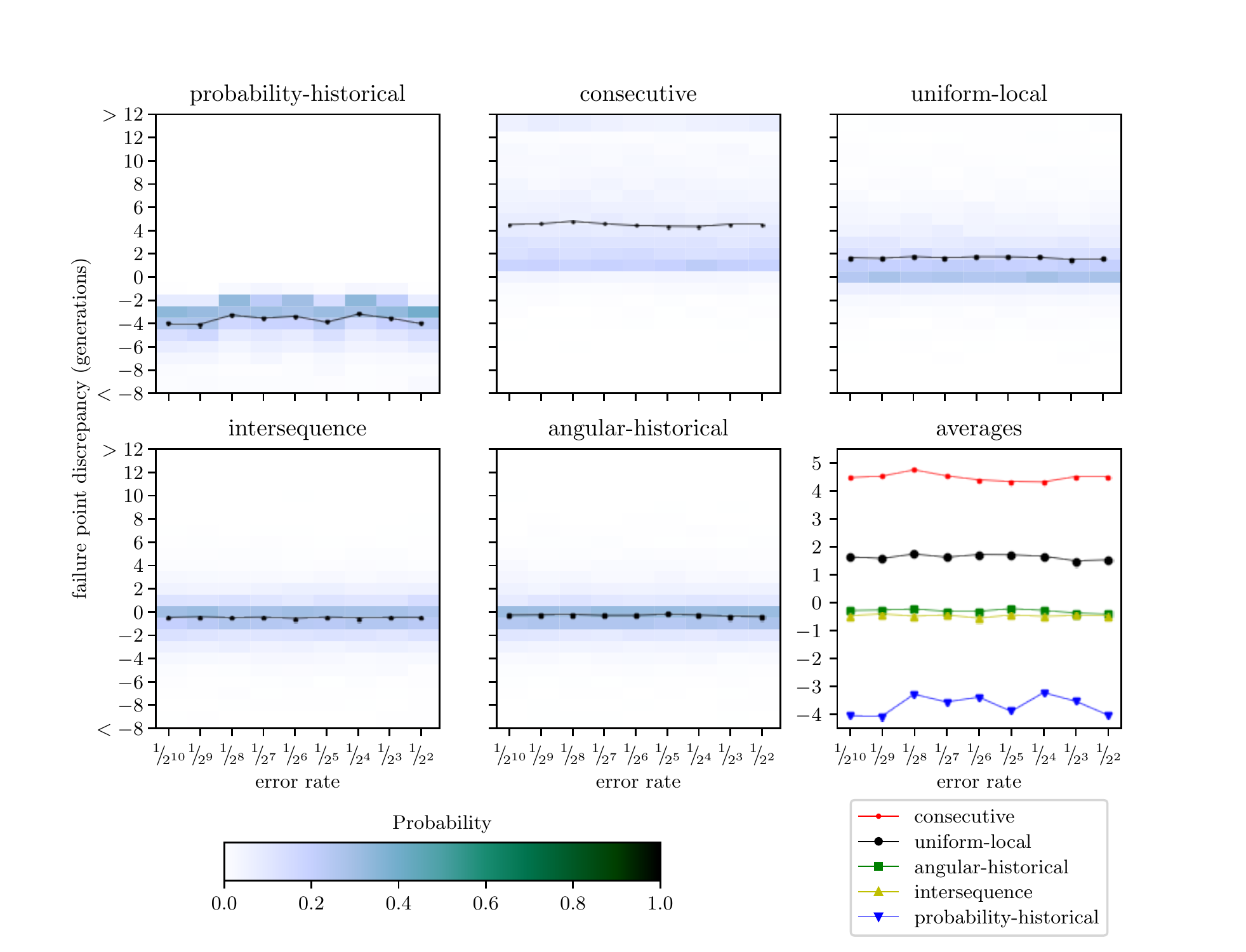}
\caption{
For varying dephasing error rates and actual angle $\theta=1.6$, plots the discrepancies between the generations at which each consistency check flagged failure, and the generations at which failure actually occurred.
Positive values indicate that the consistency check flagged failure \emph{after} failure actually occurred.
1000 RPE runs were performed per error rate.
Bin colors show the proportion of the runs at that error rate they contain.
The curves show averages for each error rate.
The lower right plot shows the average curves from all of the other plots, for comparison.
\label{dephas_error_discrepancies}
}
\end{figure*}

\begin{figure*}
\centering
\includegraphics{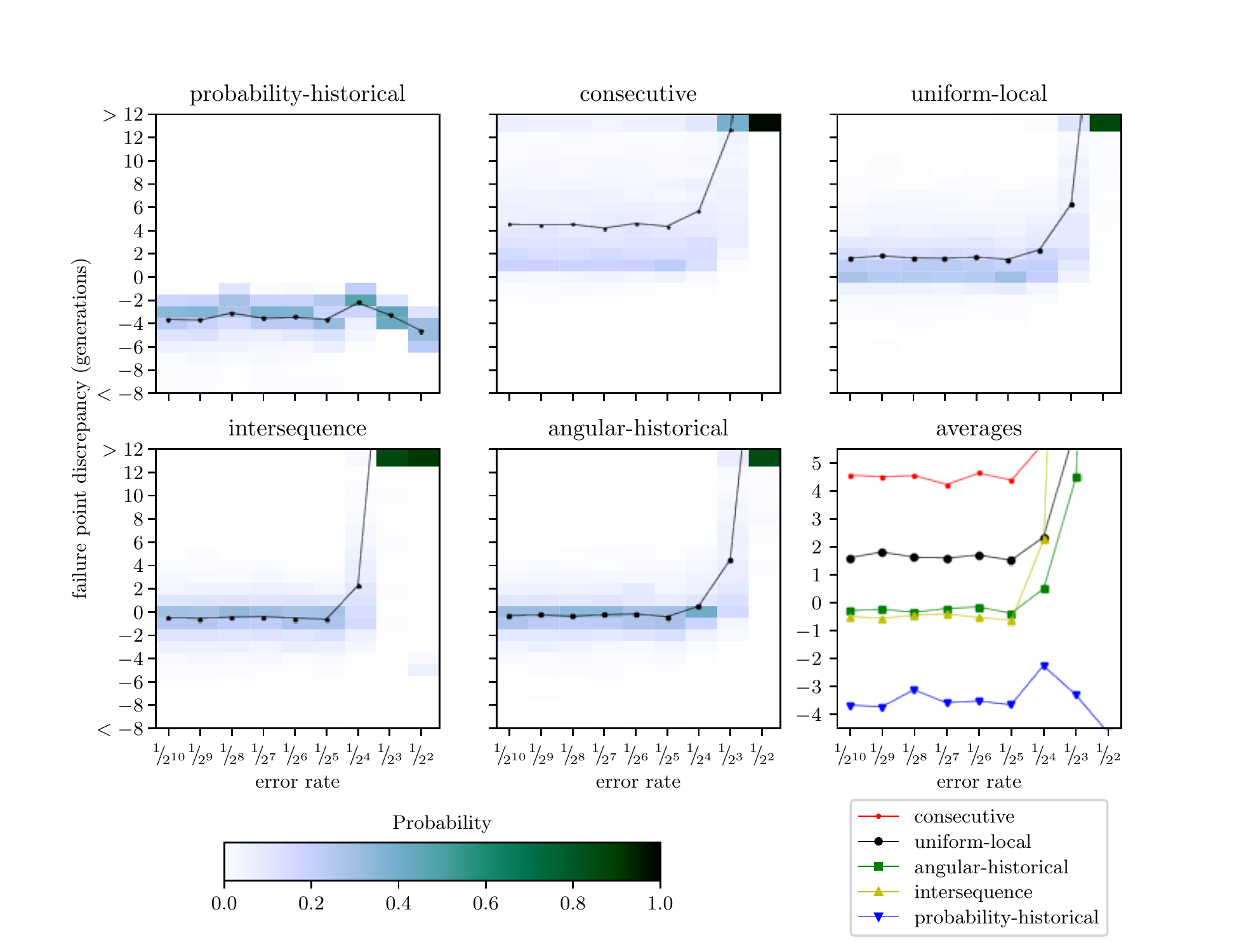}
\caption{
For varying amplitude damping error rates and actual angle $\theta=1.6$, plots the discrepancies between the generations at which each consistency check flagged failure, and the generations at which failure actually occurred.
Positive values indicate that the consistency check flagged failure \emph{after} failure actually occurred.
1000 RPE runs were performed per error rate.
Bin colors show the proportion of the runs at that error rate they contain.
The curves show averages for each error rate.
The lower right plot shows the average curves from all of the other plots, for comparison.
Very strong amplitude damping ($b\geq 1/16$) creates a false signal that RPE will track, as indicated by the residue of failures at generation $>12$.
This skews the average failure points in those generations towards the maximum probed failure generation of $45$, complicating a direct analysis of the average.
\label{amp_damp_error_discrepancies}
}
\end{figure*}

In the previous section, we defined several heuristic consistency tests for an RPE experiment.
In this section, we evaluate the performance of these tests by numerically simulating RPE runs with depolarizing, dephasing, and amplitude damping noise.

When one of our consistency tests fails, it flags a generation $k\in \{0,1,2,...,\maxk\}$ at which the RPE estimate becomes unreliable.
Because we are performing a numerical simulation for a particular target angle $\theta$, the heuristic's failure generation can be compared to the actual failure generation in which RPE is no longer able to correctly estimate $\theta$ (i.e., when it fails the condition \eqref{failure_condition}).

We can represent a single-qubit mixed state $\rho$ as a Pauli vector:
\begin{equation}
\label{rho_pauli}
    \rho=\mathds{1}+x\sigma_x+y\sigma_y+z\sigma_z
    \sim
    \begin{pmatrix}
        1\\x\\y\\z
    \end{pmatrix},
\end{equation}
for $x,y,z\in\mathbb{R}$ satisfying $x^2+y^2+z^2\le1$, and Pauli matrices $\sigma_x,\sigma_y,\sigma_z$.
Quantum operations are then implemented as superoperators that act on $\rho$ in the vector representation \eqref{rho_pauli}.

We assume that $\estunitary$ is a rotation about $\sigma_x$ by some angle $\theta$, the parameter we wish to estimate.
We represent $\estunitary$ as a superoperator:
\begin{equation}
\label{ideal_U}
    \estunitary=
    R_x(\theta)=
    \begin{pmatrix}
        1&0&0&0\\
        0&1&0&0\\
        0&0&\cos(\theta)&-\sin(\theta)\\
        0&0&\sin(\theta)&\cos(\theta)
    \end{pmatrix}.
\end{equation}
Let $\ket{0}$ and $\ket{1}$ denote the $\sigma_z$ eigenstates with eigenvalues $+1$ and $-1$, respectively.
Ideally, the initial state is $\ket{0}$:
\begin{equation}
    \rho_\text{init}=\ket{0}\bra{0}
    \sim
    \begin{pmatrix}
        1\\0\\0\\1
    \end{pmatrix},
\end{equation}
and the measurements of the cosine and sine strings are ideally of the states
\begin{equation}
    \rho_\text{c}=\ket{0}\bra{0}
    \sim
    \begin{pmatrix}
        1\\0\\0\\1
    \end{pmatrix},\quad\text{and}\quad
    \rho_\text{s}=\ket{i}\bra{i}
    \sim
    \begin{pmatrix}
        1\\0\\1\\0
    \end{pmatrix}.
\end{equation}
Denoting the realistic noisy (rather than ideal) forms of these quantities with tildes, the probabilities of the cosine and sine measurements are
\begin{equation}
    \left(\begin{array}{c}P_{\mathrm{c},k}\\P_{\mathrm{s},k}\end{array}\right)
    =\left(\begin{array}{c} \tilde\rho^\dag_\text{c} \tilde U_k \tilde\rho_\text{init}
    \\\tilde\rho^\dag_\text{s} \tilde U_k \tilde \rho_\text{init}\end{array}\right).
\end{equation}
We reparameterize these probabilities as
\begin{equation}
    \left(\begin{array}{c}2P_{\mathrm{c},k}-1\\2P_{\mathrm{s},k}-1\end{array}\right)
    =\left(\begin{array}{c}\lambda\cos\hat\phi\\\lambda\sin\hat\phi\end{array}\right),
\label{eqn:noisy_probability_form}
\end{equation}
where $\hat\phi$ is the maximum likelihood estimate of the angle $N_k\theta$, and $\lambda$ renormalizes the measurement counts $M\mapsto M\lambda^2$, as shown in \cref{app:sample_complexity}.
The renormalization occurs when a protocol makes use of only $\hat\phi$ information, so RPE and the consistency checks---with the exception of the probability-historical consistency check---all have this scaling behavior.
In order to simulate depolarization, dephasing, or amplitude damping, after each application of the unitary $\estunitary$ of \cref{ideal_U} we apply an error operator $V$ defined by the error type and the error rate $b$, i.e., $\tilde U_k=(V\estunitary)^{N_k}$.

For depolarization, the superoperator is
\begin{equation}
    V_\text{depol.}(b)=
    \begin{pmatrix}
        1&0&0&0\\
        0&1-b&0&0\\
        0&0&1-b&0\\
        0&0&0&1-b
    \end{pmatrix}.
\end{equation}
Notice that this superoperator commutes with $\estunitary$, leading to $\lambda=(1-b)^{N_k}$, while leaving $\hat\phi=N_k\theta$ unaffected.
The behavior in \cref{depol_error_discrepancies} as function of error rate is therefore equivalent to rescaled finite sample noise, so our consistency checks are also capable of catching failures due to too few shots.

Dephasing in the $\sigma_x\sigma_y$-plane\footnote{We choose to simulate dephasing noise along the $\sigma_x$- and $\sigma_y$-axes because, under the action of $\estunitary$, the state remains in the $\sigma_y\sigma_z$-plane, so dephasing along $\sigma_y$ and $\sigma_z$ would simply look like depolarization.} results in the superoperator:
\begin{equation}
    V_\text{depha.}(b)=
    \begin{pmatrix}
        1&0&0&0\\
        0&1-b&0&0\\
        0&0&1-b&0\\
        0&0&0&1
    \end{pmatrix}.
\end{equation}
In contrast to depolarizing noise, dephasing noise results in a nontrivial $\hat\phi(k)$ and $\lambda(k)$ that we do not attempt to characterize analytically.
Nevertheless, the performance of the consistency checks is still qualitatively the same, as seen in \cref{dephas_error_discrepancies}.

Finally, we simulate an amplitude damping channel, with decay from $\ket{1}$ to $\ket{0}$.
This is described by the superoperator
\begin{equation}
    V_\text{a.d.}(b)=
    \begin{pmatrix}
        1&0&0&0\\
        0&\sqrt{1-b}&0&0\\
        0&0&\sqrt{1-b}&0\\
        b&0&0&1-b
    \end{pmatrix}.
\end{equation}
Unlike the previous examples, the presence of the $b$ term in the lower-left corner of the matrix---corresponding to the relaxation to the state $\ket{0}$---acts to drive the system to a particular steady state by adding a finite term to the $z$ component of the Bloch vector at every application of $\tilde U$.
The other terms in the noise model act as damping, so for sufficiently large $N_k$, the system evolves towards a particular fixed $\lambda(\infty)$ and $\hat\phi(\infty)$, which we again do not attempt to determine analytically.
Notice that $\lambda$ increases with $b$ in this scenario, because $b$ serves as the amplitude of the driving term added at every generation.
The overall result is that, if the amplitude damping is strong enough, the statistical noise will become irrelevant at high generations, and RPE will begin to track the strong signal for $\hat\phi(\infty)$.
Hence any consistency checks that depends on $\hat\phi$ information only will never flag a failure, because without information about the true angle there is no way to tell that the false signal $\hat\phi(\infty)$ is incorrect.
This results in a residue in \cref{amp_damp_error_discrepancies} for failure discrepancies greater than $12$ for strong amplitude damping (i.e., the consistency checks flag failure more than 12 generations after failure actually occurs).
Notice that the probability-historical consistency check does detect the error because it is directly sensitive to the length of the Bloch vector $\lambda$.
For this reason, in general the probability-historical consistency test is more pessimistic than the other tests (i.e., it flags failure earlier).
In the case of strong amplitude damping, this pessimism is justified, but we suggest caution, since other error models may not cause $\lambda$ to shrink sufficiently to flag a failure.
Note that if there is reason to believe that a small $\lambda$ is an indication of overall infidelity of the system, one might also consider directly checking the magnitude of $\lambda$.

\begin{figure*}
\centering
\includegraphics{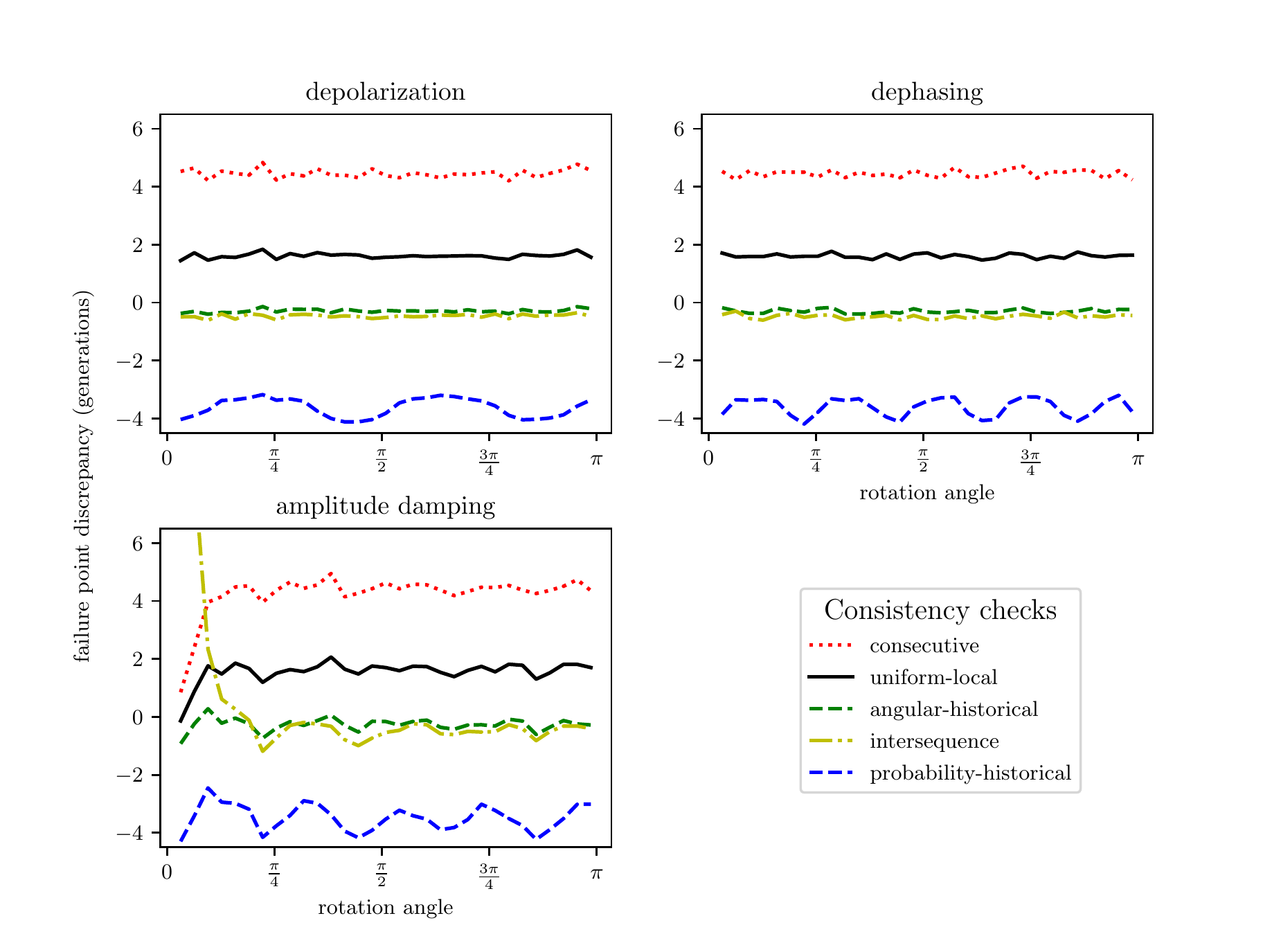}
\caption{
Dependence of the failure point discrepancies on actual rotation angle for simulated depolarization, dephasing, and amplitude damping at the rate $b=2^{-6}$.
Each data point is an average of 1000 RPE runs.
The most notable feature is that all consistency checks have near constant failure point discrepancies for all $\theta$, with the exception of substantial deviations for small $\theta$ under amplitude damping.
\label{check_vs_angle}
}
\end{figure*}

In all cases, we simulated RPE runs with each error model for exponentially spaced error rates
\begin{equation}
    b=2^{-i}\quad\text{for}\quad i=2,\ldots,10.
\end{equation}
Additionally, SPAM error was introduced by setting \begin{equation}
    \tilde\rho_\text{init.}=V_\text{depol.}(b_\text{SPAM})\rho_\text{init.},
\end{equation}
and
\begin{equation}
    \tilde\rho_\text{c}=V_\text{depol.}(b_\text{SPAM})\rho_\text{c}.
\end{equation}
Because many implementations of RPE will use an additional gate to implement $\rho_\text{s}$, we injected error due to an imperfect $\pi/2$ rotation:
\begin{equation}
\tilde\rho_\text{s}=V_\text{depol.}(b_\text{SPAM})V_\text{depol.}(b_\text{s})R_x(b_\text{s})\rho_\text{s}.
\end{equation}
We choose and present results for a fixed value of $b_\text{SPAM}=b_\text{s}=10^{-2}$, but comment that the results are qualitatively the same with both set to zero, $b_\text{SPAM}=b_\text{s}=0$.
For each error rate and error type, we simulated 1000 runs of the RPE procedure, taking $M=1000$ samples of the measurement outcomes at each generation, and repeating for a variety of angles $\theta$.
The results for
\begin{equation}
    \theta=1.6
\end{equation}
are shown in \cref{depol_error_discrepancies,dephas_error_discrepancies,amp_damp_error_discrepancies}.
The primary generation sequence in all cases was $N_k=2^k$, and for the intersequence consistency check the second sequence was $N_0=2$, $N_i=3(2^{i-1})$ for $i=1,2,3,...$ (with an initial generation at $N_{-1}=1$ that was not compared to the primary sequence).
Each plot in these figures compares the generation at which the heuristic consistency check flagged failure to the actual failure point of the run (as determined by \eqref{failure_condition}).
In particular, we subtract the generation number where failure actually occurred from the generation number that was flagged by the consistency check.
Thus positive values in \cref{depol_error_discrepancies,dephas_error_discrepancies,amp_damp_error_discrepancies} indicate that the actual failure occurred \emph{before} failure was flagged by the given consistency check.
Data was collected out to 45 generations (and $46$ for the intersequence alternative test), and if no failure was detected, a failure point at the following generation was recorded.
This choice can only affect the calculation of average failure point.
In addition, this only becomes an issue for amplitude damping with strong error rates $b\geq 1/16$.

The results in \cref{depol_error_discrepancies,dephas_error_discrepancies,amp_damp_error_discrepancies} show that the angular-historical consistency check is on average the closest to the actual failure point, in all the cases we studied.
(The angular-historical consistency check is found in the center of the bottom row of each figure.)
Close behind it is the intersequence consistency check (which appears on the left of the bottom row of each figure).
We therefore suggest that if one simply desires to estimate as accurately as possible the actual failure point, one should use the angular-historical consistency check.
If one wants an additional verification layer for the resulting failure generations, one could compare these results to those of the intersequence consistency check (which, recall, requires taking a second set of data).
Since the angular-historical and intersequence tests perform similarly for most cases of the error models studied in this paper, finding a large difference between them in an experiment would indicate that the underlying error model is outside the regimes investigated here, or is in one of the pathological cases for amplitude damping.

If one instead wants to obtain a conservative estimate of the failure point, i.e., an estimate that precedes the actual failure point with high probability, one should use the check for probability-historical consistency (which appears in the top left corner of each figure).
As we can see from \cref{depol_error_discrepancies,dephas_error_discrepancies,amp_damp_error_discrepancies}, in every run that we simulated, the check for probability-historical consistency flagged failure before failure actually occurred.

The results in \cref{depol_error_discrepancies,dephas_error_discrepancies,amp_damp_error_discrepancies} are only for $\theta=1.6$, as noted above.
However, their qualitative features appear to hold for almost any $\theta$, the exception being small $\theta$ under amplitude damping.
In particular, we see in \cref{check_vs_angle} that at the error rate $b=2^{-6}$, the angular-historical consistency check and the intersequence consistency check are the closest to correct on average, and the probability-historical consistency check flags failure early.

\begin{figure}
\centering
\includegraphics{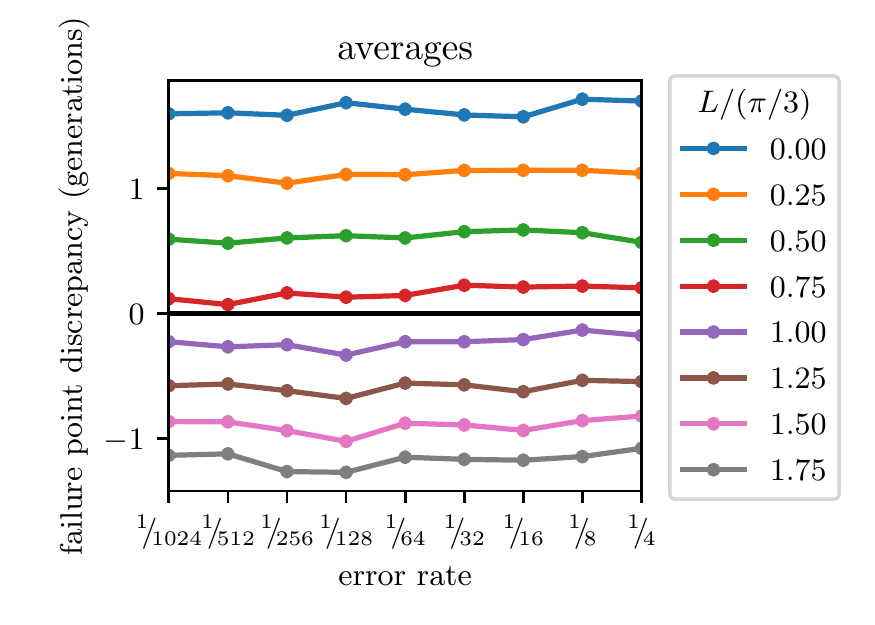}
\caption{
For varying depolarization error rates and actual angle $\theta=1.6$, plots the failure discrepancies for the angular-historical consistency check for various minimum interval widths ($L$, as in \eqref{angular_interval_consistency_eq_alt}).
Positive values indicate that the consistency check flagged failure \emph{after} failure actually occurred.
Each point is averaged over 1000 runs.
\label{depol_error_disc_min_interval}
}
\end{figure}

Testing angular-historical consistency in the case $N_k=2^k$ amounts to checking that each intersection \eqref{angular_interval_consistency_eq_alt} has size at least $\frac{\pi}{3N_k}$, i.e.,
\begin{equation}
    L=\delta\theta_k=\frac{\pi}{3}
\end{equation}
(by \eqref{delta_theta_k}).
The fact that \eqref{angular_interval_consistency_eq_alt} is equivalent to \eqref{angular_interval_consistency_eq} for $L=\delta\theta_k$ suggests that this value of $L$ should provide good performance of the consistency check, which is supported by Figs.~\ref{depol_error_discrepancies}, \ref{dephas_error_discrepancies}, and \ref{amp_damp_error_discrepancies}.

However, as discussed in the paragraph following \eqref{angular_interval_consistency_eq_alt}, we could in principle build a heuristic consistency check around any value of $L$ we like, rather than $L=\delta\theta_k$.
As for any heuristic consistency check, performance will depend on the specific error model.
Consequently, we tested the angular-historical consistency check for a variety of interval widths centered around $L=\delta\theta_k=\pi/3$; an example of the results is shown in \cref{depol_error_disc_min_interval}.
This and the plots for other actual angles show that although the angular-historical consistency check with interval width $L=\pi/3$ (as defined in \cref{unif_approx_pi_over_3}) is close to optimal, it does flag failure early on average, so it might be possible to numerically fine-tune the interval width in order to obtain a more accurate check.
However, doing so is sensitive to the specific error model and error rate as well as the actual angle, so the utility of this approach is probably limited, and we instead chose to stick to the theoretically-motivated width of $L=\pi/3$.

\ \section{Conclusion\label{sec:conclusion}}

In this work we provided a framework for characterizing the consistency of RPE data based on a variety of efficiently classically verifiable criteria.
The implementation of such consistency checks will allow an experimenter to address the worry that, due to systematic errors, their RPE run may have violated the assumptions that guarantee the protocol's performance.
Such a violation might result in the protocol returning a dramatically incorrect estimate of the value of the desired parameter.
We described seven such checks, and tested them numerically under simulated depolarization, dephasing, and amplitude damping, identifying the angular-interval-consistency check as the most accurate in all cases.
This provides a tool that augments the standard RPE protocol by permitting detection of unknown errors with unknown rates, which would otherwise cause hidden failure of the standard RPE protocol.

\begin{acknowledgments}
W. M. K. acknowledges support from the NSF, Grant No. DGE-1842474, and the NSF STAQ project, Grant No. PHY-1818914.
This work was supported in part by the U.S. Department of Energy, Office of Science, Office of Advanced Scientific Computing Research, Quantum Algorithms Team and Quantum Computing Applications Team programs.
Sandia National Laboratories is a multi-mission laboratory managed and operated by National Technology and Engineering Solutions of Sandia, LLC, a wholly owned subsidiary of Honeywell International, Inc., for DOE's National Nuclear Security Administration under contract DE-NA0003525.
\end{acknowledgments}

\bibliography{references}

\clearpage\widetext\appendix

\section{Derivation of Criteria}
\label[appendix]{app:derive-criteria}
This appendix provides proofs for some of the results in \cref{sec:criteria}, as well as some additional details.
We will extensively use the distance $d$ from an angle to a set of angles, as well as the minimizer $\M$ of $d$, so we repeat their definitions \cref{eq:distance_def_pt} and \cref{eq:minimizer_def} here.
First, we restate the point distance on the unit circle,
\begin{equation}
|\theta'-\theta''|_{2\pi} = \min\left\{\left|\theta'-\theta''+2\pi n\right|~\big| ~n\in\mathbb{Z}\right\}.
\end{equation}
which induces a set distance,
\begin{equation}
\label{eq:distance_def_set_app}
d(\theta',\Theta)=\min_{\tilde\theta\in\Theta}\left|\tilde\theta-\theta'\right|_{2\pi}.
\end{equation}
Which has the minimizer
\begin{equation}
\label{eq:minimizer_def_app}
\M(\theta',\Theta) = \argmin_{\tilde\theta\in\Theta}\left|\tilde\theta-\theta'\right|_{2\pi}.
\end{equation}
We will additionally use a version of the triangle inequality appropriate for point-to-set distances.
\begin{theoremnonum}[Set Triangle Inequality]
$d(\theta,\Theta)\leq |\theta-\theta'|_{2\pi}+d(\theta',\Theta)$
\end{theoremnonum}
\begin{proof}
\ \\
\begin{equation}
d(\theta,\Theta)=\min_{\tilde\theta\in\Theta}\left|\tilde\theta-\theta\right|_{2\pi}\leq \min_{\tilde\theta\in\Theta}\left(\left|\tilde\theta-\theta'\right|_{2\pi}+\left|\theta'-\theta\right|_{2\pi}\right)=\left|\theta'-\theta\right|_{2\pi}+\min_{\tilde\theta\in\Theta}\left|\tilde\theta-\theta'\right|_{2\pi}=\left|\theta'-\theta\right|_{2\pi}+d(\theta',\Theta)
\end{equation}
\end{proof}
\subsection{Plausible Consistency}
Recall our definition in \cref{X_k_def} of the plausible consistent angles at each generation $k$,
\begin{equation}
\label{eq:plau_def_app}
\Plau_k=\left\{\tilde{\theta}~\middle|~\hat{\theta}_k=\M(\tilde{\theta},\Theta_k)\right\}.
\end{equation}
We call $\Plau_k$ the plausible angles for generation $k$ because the estimate $\hat\theta_k$ chosen at generation $k$ is ``correct" (as defined by \cref{failure_condition}) if and only if the actual angle is in $\Plau_k$.
Therefore, if the actual angle were any $\tilde\theta \in \bigcap_{k\leq \maxk} \Plau_k$, then the entire sequence of estimates $\{\hat\theta_{k}\}_{k\leq\maxk}$ would be correct.
Hence, the corresponding consistency check is
\begin{equation}
\label{eq:plaus_crit_app}
    \bigcap_{k=0}^\maxk \Plau_k \neq \emptyset.
\end{equation}
\cref{eq:plaus_crit_app} is the same as \cref{plausible_consistency}, in the main text.
\begin{remarknote}\label{app:plau_def2}
The plausible angles at generation $k$ are exactly those angles for which distance to the set $\Theta_k$ is the same as the distance to the RPE-chosen angle $\hat\theta_k\in\Theta_k$,
\begin{equation}
\Plau_k=\left\{\tilde{\theta}~\middle|~|\hat{\theta}_k-\tilde\theta|_{2\pi}=d(\tilde{\theta},\Theta_k)\right\}.
\end{equation}
\end{remarknote}
\begin{remarknote}
The angles within $\pi/N_k$ of the RPE-chosen angle $\hat\theta_k$ are the plausible angles at generation $k$,
\begin{equation}
\Plau_k=\left\{\tilde{\theta}~\middle|~|\tilde\theta-\hat\theta_{k}|_{2\pi}<\frac{\pi}{N_{k'}}\right\}
\end{equation}
\end{remarknote}
We have used intervals that are open on both the left and right sides because we have not chosen a convention to use, e.g., the angle that is closer to $0$ when restricted to the principal range $[0,2\pi)$.
While this introduces another failure mode to the analysis, albeit a low-probability one, we will see in \cref{well_approx_thm} that this failure mode is ruled out by consecutive-consistency.

\begin{remarknote}\label{app:remark_X_k_subset}
If the plausible consistency check \cref{eq:plaus_crit_app} is satisfied as $\maxk\to\infty$, we are able to fully resolve the angle.
In other words, if there exists a $\tilde\theta\in\bigcap_{k<\infty} \Plau_k$, then $\tilde\theta$ is unique and $\lim_{k\to\infty}\hat\theta_k\to\tilde\theta$.
\end{remarknote}

However, as noted in the main text, the implications of \cref{eq:plaus_crit_app} for finitely many generations of data are limited in some cases.
\begin{remarknote}
If $N_{k}\ge2N_{k-1}$ for all $k$, then $\Plau_{k}\subset \Plau_{k-1}$.
Equivalently, $\bigcap_{k\le \maxk} \Plau_{k}=\Plau_\maxk$, and the consistency criterion \cref{eq:plaus_crit_app} is always satisfied.
\end{remarknote}
\begin{proof}
Assume that $N_k\ge2N_{k-1}$, and let $\tilde\theta\in \Plau_k$.
We prove that $\tilde\theta\in\Plau_{k-1}$.
First, notice that
\begin{equation}
\label{remark2_step1}
    |\hat\theta_{k-1}-\hat\theta_k|_{2\pi} \leq \pi/N_k,
\end{equation}
since $\hat\theta_k$ is chosen to be the element of $\Theta_k$ that is closest to $\hat\theta_{k-1}$, and $\Theta_k$ is made up of $N_k$ equally-spaced angles in $[0,2\pi)$.
Then, using the triangle inequality,
\begin{equation}
\begin{split}
    |\hat{\theta}_{k-1}-\tilde\theta|_{2\pi}&\leq
    |\hat\theta_{k-1}-\hat\theta_k|_{2\pi}+|\hat\theta_k-\tilde\theta|_{2\pi}\\
    &\leq \frac{\pi}{N_k}+\frac{\pi}{N_k}\\
    &= \frac{\pi}{N_k/2}\leq\frac{\pi}{N_{k-1}},
\end{split}
\end{equation}
where the second line follows from \cref{remark2_step1} and from the definition \cref{eq:plau_def_app} of $\Plau_k$.
Hence by \cref{eq:plau_def_app}, $\tilde\theta\in\Plau_{k-1}$, and thus $\Plau_k\subset\Plau_{k-1}$.
\end{proof}

\subsection{Consecutive Consistency}
Recall the definition of consecutive consistency in \cref{consec_crit} (\cref{eq:cons_crit}, in the main text),
\begin{equation}
\label{eq:cons_crit_app}
    \bigcap_{k=0}^\maxk \Cons_k \neq \emptyset,
\end{equation}
where
\begin{equation}
\label{eq:cons_def_app}
    \Cons_k=\left\{\tilde{\theta}~\middle|~d(\tilde\theta,\Theta_k)+d(\tilde\theta,\Theta_{k-1})<\frac{\pi}{N_k}\right\}
\end{equation}
for $k\in \{1,2,\ldots,\maxk\}$, and $\Cons_0=[0,2\pi)$.
Notice first that consecutive consistency does not make reference to the RPE-chosen angles $\hat\theta_k$.
That \cref{consec_crit} (\cref{eq:cons_crit_app}) is stronger than \cref{plausible_crit} (\cref{eq:plaus_crit_app}) is proven in the following theorem:

\textbf{Theorem \ref{well_approx_thm}.}
\emph{
The sets $\lbrace \Plau_{k} \rbrace_{k\leq \maxk}$ and $\lbrace \Cons_{k} \rbrace_{k\leq \maxk}$ satisfy:
\begin{equation}
\bigcap_{k'\leq k} \Cons_{k'} \subseteq \bigcap_{k'\leq k} \Plau_{k'}.
\end{equation}
Hence, \cref{eq:cons_crit} implies \cref{plausible_consistency}.
}
\begin{proof}
We will need to show along the way that each $\hat\theta_k=\M(\hat\theta_{k-1},\Theta_{k})$ is well-defined---this may fail to be the case if the choice is between two candidate angles that are equidistant to the previous $\hat\theta_k$ (i.e., $d(\hat\theta_{k-1},\Theta_k)=\pi/N_k$).
We therefore proceed by induction on $k$ for the following statement: there is a unique minimizer $\M(\hat\theta_{k-1},\Theta_{k})$ (provided $k>0$), and
\begin{equation}
\label{eq:induction_assumption}
\bigcap_{k'\le k} \Cons_{k'} \subseteq \bigcap_{k'\le k} \Plau_{k'}.
\end{equation}
The base case is $k=0$, where \cref{eq:induction_assumption} becomes $\Cons_{0}\subseteq\Plau_{0}$, which follows directly from the definitions \cref{eq:plau_def_app} and \cref{eq:cons_def_app}.

For the induction step, assume that for some $k$ the inductive hypothesis (i.e., \cref{eq:induction_assumption} and uniqueness of the minimizer) holds.
Then for any $\tilde\theta\in \bigcap_{k'\leq k+1} \Cons_{k'}$ we have $\tilde\theta\in \Plau_{k}$, and (by \cref{app:plau_def2}) $|\hat\theta_{k}-\tilde\theta|_{2\pi} = d(\tilde\theta,\Theta_{k})$.
Therefore by the triangle inequality,
\begin{equation}
\label{eq:intermediate_step}
d(\hat\theta_{k},\Theta_{k+1})\leq |\hat\theta_{k}-\tilde\theta|_{2\pi}+d(\tilde\theta,\Theta_{k+1}) = d(\tilde\theta,\Theta_{k})+d(\tilde\theta,\Theta_{k+1}) < \frac{\pi}{N_{k+1}},
\end{equation}
where the final inequality follows from the definition of $\Lambda_{k+1}$.
Because the final inequality in \cref{eq:intermediate_step} is strict, the minimizer $\M(\hat\theta_{k},\Theta_{k+1})$ is unique.
Next, set
\begin{equation}
B=\bigcap_{k'\le k} \Cons_{k'} \subseteq \bigcap_{k'\le k} \Plau_{k'}.
\end{equation}
We will show that
\begin{equation}
B\cap \Cons_{k+1} \subseteq B\cap \Plau_{k+1}.
\end{equation}
Let $\tilde\theta\in B\cap \Cons_{k+1}$.
Then we have
\begin{equation}
\label{eq:final_step}
\begin{split}
\left| \M(\tilde\theta,\Theta_{k+1})-\hat\theta_{k}\right|_{2\pi} &\leq \left|\M(\tilde\theta,\Theta_{k+1})-\tilde\theta\right|_{2\pi} + \left|\tilde\theta-\hat\theta_{k}\right|_{2\pi}\\
& = \left|\M(\tilde\theta,\Theta_{k+1})-\tilde\theta\right|_{2\pi} + \left|\tilde\theta-\M(\tilde\theta,\Theta_{k})\right|_{2\pi}\\
&= d(\tilde\theta,\Theta_{k+1})+d(\tilde\theta,\Theta_{k})\\
&< \frac{\pi}{N_{k+1}},
\end{split}
\end{equation}
where the first line follows by the triangle inequality, the second line follows because $\tilde\theta\in\Plau_{k}$ and thus $\hat\theta_{k}=\M(\tilde\theta,\Theta_{k})$, the third line follows by the definition \cref{eq:distance_def_set_app} of the distance $d$, and the final inequality follows because $\tilde\theta\in\Cons_{k+1}$.
Therefore, since the elements of $\Theta_{k+1}$ are separated by $2\pi/N_{k+1}$, by \cref{eq:final_step} $\M(\tilde\theta,\Theta_{k+1})$ must be the closest element in $\Theta_{k+1}$ to $\hat\theta_k$.
But that closest element is $\hat\theta_{k+1}$, by definition, so $\hat\theta_{k+1}=\M(\tilde\theta,\Theta_{k+1})$, i.e., $\tilde\theta\in \Plau_{k+1}$.
\end{proof}

\begin{corollary}
\label{app_cor:well_approx-stable}
If $\tilde\theta\in \bigcap_{k'\leq k} \Cons_{k'}$, and $\left\{ \Theta'_{k'}\right\}_{k'\leq\maxk}$ are another set of measurement data satisfying
\begin{equation}
 d(\tilde\theta,\Theta_{k'}') \leq d(\tilde\theta,\Theta_{k'}),
\end{equation}
then $\tilde\theta\in \bigcap_{k'\leq k} \Cons'_{k'}\subseteq \bigcap_{k'\leq k} \Plau'_{k'}$, where $\left\{ \Cons'_{k'}\right\}_{k'\leq k}$ and $\left\{ \Plau'_{k'}\right\}_{k'\leq k}$ are generated by the $\Theta'$.
\end{corollary}
The corollary shows that any improvement of the measurements that reduces the error to any of the consecutively consistent values $\tilde\theta\in \bigcap_{k'\leq \maxk} \Cons_{k'}$ will still cause that $\tilde\theta$ to be identified as a correct value (c.f. \cref{fig:pathology_Xk}, illustrating that this is \emph{not} the case for $\bigcap_{k'\leq \maxk} \Plau_{k'}$).

Notice that the $\Plau_k$ sets are defined in terms of distance from the $\hat\theta_k$, automatically making them intervals.
The $\Cons_k$ are instead defined in terms of distance to the $\Theta_k$.
\cref{well_approx_thm} and \cref{app:plau_def2} together imply that $\bigcap_{k'\leq k}\Cons_{k'}$ is a simple interval.
\begin{lemma}
\label{app:consecutive_consistent_interval}
Let $D=\frac{\pi}{2N_k}-\frac{1}{2}\left|\hat\theta_k-\hat\theta_{k-1}\right|_{2\pi}$.
If $D<0$, then $\Plau_k\cap \Cons_k = \emptyset$.
Otherwise,
\begin{equation}
\label{eq:lemma1_claim}
\Cons_k\cap \Plau_k= \left(\hat\theta_{k-1}-D,\hat\theta_k+D\right)_{2\pi}.
\end{equation}
The $2\pi$ subscript indicates that the interval is circular and should be interpreted as follows: connect $\hat\theta_k$ and $\hat\theta_{k-1}$ along the shortest arc.
Expand that arc by a circular distance $D$ on both sides.
Also, the arc $\Cons_k\cap \Plau_k$ has length $<\pi$.
\end{lemma}

\begin{proof}
We first prove that the arc $\Cons_k\cap \Plau_k$ has length $<\pi$.
The worst case is when $k=1$: in this case, assuming that $N_k$ is strictly increasing, $N_1\ge2$, so $D\le\pi/4$.
Since $\Theta_1$ contains $N_k\ge2$ angles, $|\hat\theta_1-\hat\theta_0|_{2\pi}<\pi/2$.
Hence the arc $\left(\hat\theta_{k-1}-D,\hat\theta_k+D\right)_{2\pi}$ has length at most $\pi$.

For the main proof, first note that by applying \cref{app:plau_def2} to the distances in the definition of $\Cons_k$ (\cref{eq:cons_def_app}),
\begin{equation}
\Plau_k\cap \Cons_k = \Plau_k \cap\left\{ \tilde\theta\ \middle|\ \left| \tilde\theta-\hat\theta_k\right|_{2\pi}+\left| \tilde\theta-\hat\theta_{k-1}\right|_{2\pi} < \frac{\pi}{N_k} \right\}.
\end{equation}
Because of the interval representation of $\Plau_k$,
\begin{equation}
\Plau_k\cap \Cons_k = (\hat\theta_k-\pi/N_k,\hat\theta_k+\pi/N_k)_{2\pi}\cap\left\{ \tilde\theta\ \middle|\ \left| \tilde\theta-\hat\theta_k\right|_{2\pi}+\left| \tilde\theta-\hat\theta_{k-1}\right|_{2\pi} < \frac{\pi}{N_k} \right\}.
\end{equation}
The first interval is a superset of the second set therefore
\begin{equation}
\label{eq:lemma1_set_rep}
\Plau_k\cap \Cons_k = \left\{ \tilde\theta\ \middle|\ \left| \tilde\theta-\hat\theta_k\right|_{2\pi}+\left| \tilde\theta-\hat\theta_{k-1}\right|_{2\pi} < \frac{\pi}{N_k} \right\}.
\end{equation}
Because $|\hat\theta_{k-1}-\hat\theta_{k}|_{2\pi}<\pi$, for any $\tilde\theta$
\begin{equation}
\label{eq:consec_cases}
\left| \tilde\theta-\hat\theta_k\right|_{2\pi}+\left| \tilde\theta-\hat\theta_{k-1}\right|_{2\pi} = |\hat\theta_{k-1}-\hat\theta_{k}|_{2\pi} + \begin{cases}
0,& \tilde\theta\in(\hat\theta_{k-1},\hat\theta_k)_{2\pi}~; \\2\min\{|\hat\theta_{k-1}-\tilde\theta|_{2\pi}, |\hat\theta_{k}-\tilde\theta|_{2\pi}\},& \tilde\theta\notin(\hat\theta_{k-1},\hat\theta_k)_{2\pi}~.
\end{cases}
\end{equation}
Consider the first case in \cref{eq:consec_cases}, where $\tilde\theta\in(\hat\theta_{k-1},\hat\theta_k)_{2\pi}$: from \cref{eq:lemma1_set_rep} we obtain
\begin{equation}
\tilde\theta\in\Plau_k\cap\Cons_k\quad\Leftrightarrow\quad\left| \hat\theta_k-\hat\theta_{k-1}\right|_{2\pi} < \frac{\pi}{N_k}.
\end{equation}
The right-hand inequality may be rewritten in terms of $D$ as
\begin{equation}
0< \frac{\pi}{2N_k} - \frac{1}{2}\left| \hat\theta_k-\hat\theta_{k-1}\right|_{2\pi} = D,
\end{equation}
and thus $\tilde\theta$ is in the interval in \cref{eq:lemma1_claim}.

Alternatively, in the second case of \cref{eq:consec_cases}, where $\tilde\theta\notin(\hat\theta_{k-1},\hat\theta_k)_{2\pi}$, $\tilde\theta\in\Plau_k\cap\Cons_k$ is equivalent to
\begin{equation}
\left| \hat\theta_k-\hat\theta_{k-1}\right|_{2\pi} +2\min\{|\hat\theta_{k-1}-\tilde\theta|_{2\pi}, |\hat\theta_{k}-\tilde\theta|_{2\pi}\} < \frac{\pi}{N_k},
\end{equation}
or
\begin{equation}
\min\{|\hat\theta_{k-1}-\tilde\theta|_{2\pi}, |\hat\theta_{k}-\tilde\theta|_{2\pi}\}< \frac{\pi}{2N_k} - \frac{1}{2}\left| \hat\theta_k-\hat\theta_{k-1}\right|_{2\pi} = D,
\end{equation}
which is precisely checking if $\tilde\theta$ is within $D$ of the endpoints of the angular interval, also precisely as desired.
\end{proof}

\begin{theorem}
Testing for membership in $\bigcap_{k'<k} \Cons_{k'}$ is equivalent to testing for membership in the intersection of the intervals in \cref{app:consecutive_consistent_interval}.
\end{theorem}
\begin{proof}
This follows from \cref{app:consecutive_consistent_interval}, noticing that, by \cref{well_approx_thm}, $\bigcap_{k'\leq k} \Cons_{k'} \subseteq \bigcap_{k'\leq k}\Plau_{k'}$.
\end{proof}

\subsection{Historical Consistency}
Recall the definition of $\Delta_k[\delta\theta_{k}]$ (\cref{eq:delta_theta_set} in the main text):
\begin{equation}
\label{eq:delta_theta_set_app}
\Delta_k[\delta\theta_{k}] = \left\{\tilde \theta~\middle|~d(\tilde \theta,\Theta_k) < \frac{\delta\theta_k}{N_k}\right\},
\end{equation}
where $\{\delta\theta_k\}_{k\leq\maxk}$ is a sequence of positive real numbers.
In the following, $|A|$ denotes the length of the interval, $A$.

\begin{lemma}\label{app:overlap_interval_lemma}
Suppose $a\in\mathbb{R}$, and $0\leq L'<L$.
Then,
\begin{equation}
b\in(a,a+L) \iff |(b-L',b+L')\cap(a,a+L)| > L'.
\end{equation}
This generalizes to angular intervals as long as $L+2L'< 2\pi$.
\end{lemma}

\begin{theorem}
\label{thm:ang_hist_equiv}
Assume $\frac{\delta\theta_{k-1}}{N_{k-1}}> \frac{\delta\theta_{k}}{N_{k}}$ for all $0<k\leq \maxk$.
Then the two following statements are equivalent:
\begin{equation}
\label{thm:ang_hist:eq_hyp_contain}
\hat\theta_k \in \bigcap_{k'\leq k} \Delta_{k'}\quad \forall k\leq \maxk,
\end{equation}
and
\begin{equation}
\label{thm:ang_hist:eq_hyp_length}
\left|\bigcap_{k'\leq k} \Delta_{k'}\right|>\frac{\delta\theta_k}{N_k} \quad \forall k\leq \maxk.
\end{equation}
\end{theorem}
\begin{proof}
We proceed by induction on $\maxk$.
For the base case, note that $\hat\theta_0\in\Delta_0[\delta\theta_0]$ (i.e., \cref{thm:ang_hist:eq_hyp_contain} holds), and $N_0=1$, so $|\Delta_0[\delta\theta_0]| = 2 \delta\theta_0> \delta\theta_0$ (i.e., \cref{thm:ang_hist:eq_hyp_length} holds).

For the induction step, let $\maxk$ be a positive integer.
It suffices to assume (as an induction hypothesis) that both \cref{thm:ang_hist:eq_hyp_contain} and \cref{thm:ang_hist:eq_hyp_length} hold for $\maxk$, and prove that, under this assumption,
\begin{equation}\label{thm:ang_hist_equiv_mainpoint}
\hat\theta_{\maxk+1} \in \bigcap_{k'\leq {\maxk+1}} \Delta_{k'} \quad\Leftrightarrow\quad \left|\bigcap_{k'\leq {\maxk+1}} \Delta_{k'}\right|>\frac{\delta\theta_{\maxk+1}}{N_{\maxk+1}}.
\end{equation}
First, notice that by the definitions \cref{eq:plau_def_app} and \cref{eq:delta_theta_set_app},
\begin{equation}
\bigcap_{k'\leq {\maxk+1}} \Delta_{k'} = \left( \Plau_k\cap\Delta_{\maxk}\right) \cap \bigcap_{k'\leq {\maxk}} \Delta_{k'},
\end{equation}
and
\begin{equation}
\label{eq:b_interval}
\Plau_k\cap\Delta_{\maxk}=\left(\hat\theta_{\maxk+1}-\frac{\delta\theta_{\maxk+1}}{N_{\maxk+1}}, \hat\theta_{\maxk+1}+\frac{\delta\theta_{\maxk+1}}{N_{\maxk+1}}\right).
\end{equation}
\cref{thm:ang_hist_equiv_mainpoint} is then equivalent to
\begin{equation}
\hat\theta_{\maxk+1} \in \bigcap_{k'\leq {\maxk}} \Delta_{k'} \iff \left|\left( \Plau_k\cap\Delta_{\maxk}\right) \cap \bigcap_{k'\leq {\maxk}} \Delta_{k'}\right|>\frac{\delta\theta_{\maxk+1}}{N_{\maxk+1}}.
\end{equation}
Because $\bigcap_{k'\leq {\maxk}} \Delta_{k'}$ is an interval of length $L>\frac{\delta\theta_{\maxk+1}}{N_{\maxk+1}}$, inserting \cref{eq:b_interval} and applying \cref{app:overlap_interval_lemma} gives the claim.
\end{proof}

\section{Robust Resource Scaling}
\label[appendix]{app:heisenberg}

In Ref~\cite{Kimmel2015Dec}, RPE was shown to be robust to additive errors, which in turn makes the protocol robust to a range of physical errors including SPAM errors.
While there was found to be an mistake in the details of that analysis \cite{belliardo20a}, here we show that the big picture result still holds, and any protocol that has RPE-like characteristics can be made robust.

Let $X=\{X_1,\dots,X_\maxk\}$ be a set of binomial random variables, where it requires cost $c_k$ to obtain a sample of $X_k$.
Suppose there is a protocol that, for each $i\in[K]$, takes $m_k$ samples of $X_k$ to create an estimate $\hat{x}_k$ of $\mathbb{E}[X_k]$ (the average value of $X_k$), where $|\mathbb{E}[X_k]-\hat{x}_k|< \delta$ with probability at least $1-2\exp{[-2m_k\delta]}$.
The cost of this protocol is $\sum_{i=1}^Kc_km_k$.
(Here $\delta$ is the same constant for all $k$.)
Note that $1-2\exp{(-2m_k\delta)}$ probability of success is natural because many results that involve bounding the success probability of binomial random variables rely on Hoeffding's inequality, which produces this term.

Given such a protocol, we can simulate it using binomial random variables $X'=\{X_1',\dots,X_\maxk'\}$ that approximate $X$, if we are promised that for all $k$, $\mathbb{E}[X_k]-\mathbb{E}[X_k']<\epsilon<\delta$ for some constant $\epsilon$.
If the cost of sampling $X_k'$ is $c_k$, then the cost of the new protocol will be only a constant factor more than the original protocol.
Consider taking $m_k'=m_k\delta/(\delta-\epsilon)$ samples of $X_k'$.
Then using the Hoeffding inequality for the binomial distribution, we can obtain an estimate $\hat{x}_k'$ of $\mathbb{E}[X_k']$ to within additive error $\delta-\epsilon$ with probability of error at most
\begin{equation}
2\exp{[-2m_k'(\delta-\epsilon)]}=2\exp{[-2m_k\delta]}.
\end{equation}
Because $\mathbb{E}[X_k]-\mathbb{E}[X_k']<\epsilon<\delta$, this estimate $\hat{x}_k'$ is actually within $\delta$ of $\mathbb{E}[X_k]$
with probability of error $2\exp{(-2m_k\delta)}$.
Thus we can use our estimates $\hat{x}_k'$ in place of $\hat{x}_k$ in the original protocol and achieve the same result.
The cost is $\sum_{i=1}^Kc_km_k\delta/(\delta-\epsilon)$, as claimed.

The consequence of this analysis is that \textit{any} experiment dealing with binomial random variables that does not require precise estimates of any single variable will still be successful even if those variables become biased, at the cost of a multiplicative, constant overhead.
In particular, this means that it is possible to still achieve Heisenberg scaling using the phase estimation protocol outlined here, even in the presence of noise, as long as the noise does not shift the probabilities of the measurement outcomes by more than a constant.

However, this statement is difficult to take advantage of in practice, since knowing the how much to increase the sample number requires knowing the size of $\epsilon$.
This brings us back to the main purpose of the present work, which is to detect when the noise does not satisfy this property.

\section{Sample Complexity\label{app:sample_complexity}}

In this appendix we demonstrate the scaling of sample complexity, as a function of noise in the quantum channel, to achieve a particular target error bound.

\subsection{Preliminaries}
For sufficiently large number of samples $M$, the binomial distribution is approximately the normal distribution:
\begin{equation}
\label{binom_approx_norm}
    \Binom_{p, M}(k) \approx \Norm_{Mp,Mp(1-p)}(k)=\frac{1}{\sqrt{2\pi}\sqrt{Mp(1-p)}}e^{-\frac{1}{2}\frac{(k-Mp)^2}{Mp(1-p)}}.
\end{equation}

Also, the product of two normal distributions is a rescaled normal distribution.
We prove this here, and obtain the rescaling factor.
First, observe that
\begin{align}
    \left(2\pi\sigma_a\sigma_b\Norm_{\mu_a,\sigma_a}(x)\Norm_{\mu_b,\sigma_b}(x)\right)^{-2}
    \hspace*{-14em}&\hspace*{14em}= \exp\left[ \left(\frac{x-\mu_a}{\sigma_a}\right)^2 + \left(\frac{x-\mu_b}{\sigma_b}\right)^2\right]\\
    &=\exp\left[ \frac{x^2-2\mu_a x + \mu_a^2}{\sigma_a^2}
    + \frac{x^2-2\mu_bx + \mu_b^2}{\sigma_b^2}\right]\\
    &=\exp\left[ \left(\sqrt{\sigma_a^{-2}+\sigma_b^{-2}}x - \frac{\mu_a\sigma_a^{-2}+\mu_b\sigma_b^{-2}}{\sqrt{\sigma_a^{-2}+\sigma_b^{-2}}}\right)^2
 + \left(\mu_a^2\sigma_a^{-2} +\mu_b^2\sigma_b^{-2}-\frac{(\mu_a\sigma_a^{-2}+\mu_b\sigma_b^{-2})^2}{\sigma_a^{-2}+\sigma_b^{-2}}\right)\right]\\
    &=\left(\sqrt{2\pi}\sigma\Norm_{\mu,\sigma}(x)\right)^{-2}\exp\left[
    \mu_a^2\sigma_a^{-2} +\mu_b^2\sigma_b^{-2}-\frac{(\mu_a\sigma_a^{-2}+\mu_b\sigma_b^{-2})^2}{\sigma_a^{-2}+\sigma_b^{-2}}\right]
\end{align}
where $\sigma^{-2}=\sigma_a^{-2}+\sigma_b^{-2}$ and $\mu=\sigma^2(\mu_a\sigma_a^{-2}+\mu_b\sigma_b^{-2})$.
It follows that
\begin{equation}
\Norm_{\mu_a,\sigma_a}(x)\Norm_{\mu_b,\sigma_b}(x)=\Norm_{\mu,\sigma}(x) \frac{\sigma}{\sqrt{2\pi}\sigma_a\sigma_b}\exp\left[-\frac{1}{2}\left(\mu_a^2\sigma_a^{-2} +\mu_b^2\sigma_b^{-2}-\frac{(\mu_a\sigma_a^{-2}+\mu_b\sigma_b^{-2})^2}{\sigma_a^{-2}+\sigma_b^{-2}}\right)\right].
\end{equation}
Better yet, the scale factor can be expressed in terms of another normal
distribution:
\begin{equation}
    \Norm_{\mu_a,\sigma_a}(x)\Norm_{\mu_b,\sigma_b}(x)=\Norm_{\mu,\sigma}(x)\Norm_{0,\sigma'}(x'),\label{app:norm_pdf_prod}
\end{equation}
where $x'=\mu_a-\mu_b$ and $\sigma'^2=\sigma_a^2+\sigma_b^2$.

\subsection{Complexity}

The experimental data is used by the RPE algorithm to generate an estimate of the angle of rotation of $\estunitary$.
As described in \cref{eq:candidate_estimates}, this is the angle $\hat\phi$ satisfying
\begin{equation}
    \frac{2\frac{\hat s}{M}-1}{2\frac{\hat c}{M} -1} = \tan\hat\phi.\label{app:eqn_phi_constraint}
\end{equation}
where $\hat c$ and $\hat s$ are the sample counts from \cref{eq:cosine_distribution} and \cref{eq:sine_distribution}, respectively, and $M$ is the total number of measurements.
It follows that the probability of measuring an angle $\hat \phi$ is
\begin{equation}
\sum_{t,t'=0}^M \Binom_{{P}_{\mathrm{c},k}, M}\left(t \right)
    \Binom_{{P}_{\mathrm{s},k}, M}\left(t'\right) \delta_{F(t,t')},
\end{equation}
where $\delta$ is the Kronecker delta function and
\begin{equation}
F(t,t') = \cot\hat\phi\left( \frac{t}{M}-\frac{1}{2}\right) -\left(\frac{t'}{M} -\frac{1}{2}\right).
\end{equation}
For sufficiently large $M$, using \eqref{binom_approx_norm} we can approximate the probability density for calculating $\hat\phi$ from the count data as:
\begin{align}
    &\sum_{t,t'=0}^M \Norm_{M{P}_{\mathrm{c},k},M{P}_{\mathrm{c},k}(1-{P}_{\mathrm{c},k})}(t)
    \Norm_{M{P}_{\mathrm{s},k},M{P}_{\mathrm{s},k}(1-{P}_{\mathrm{s},k})}(t') \delta_{F(t,t')},\\
    \label{line2c12}
    \approx&\ \iint\,dt\,dt'\,
    \Norm_{M{P}_{\mathrm{c},k},M{P}_{\mathrm{c},k}(1-{P}_{\mathrm{c},k})}(t)
    \Norm_{M{P}_{\mathrm{s},k},M{P}_{\mathrm{s},k}(1-{P}_{\mathrm{s},k})}(t')\delta(AF(t,t'))\\
    =\ &\frac{1}{M^2}\iint\,dt\,dt'\,
    \Norm_{{P}_{\mathrm{c},k},\sigma_{\mathrm{c}}^2/M}(t/M)
    \Norm_{{P}_{\mathrm{s},k},\sigma_s^2/M}(t'/M)\delta(AF(t,t')),\label{c13}
\end{align}
where
\begin{equation}
    \sigma_s^2 = {P}_{\mathrm{s},k}(1-{P}_{\mathrm{s},k}),\quad\text{ and }\quad
    \sigma_{\mathrm{c}}^2 = {P}_{\mathrm{c},k}(1-{P}_{\mathrm{c},k}),
\end{equation}
and in \eqref{line2c12} we have changed the Kronecker delta to a Dirac delta and introduced a free parameter $A$ to ensure proper normalization.
We further simplify this expression by parameterizing the probability using \cref{eqn:noisy_probability_form},
\begin{equation}
\label{prob_params}
    \left(\begin{array}{c}2P_{\mathrm{c},k}-1\\2P_{\mathrm{s},k}-1\end{array}\right)
    =\left(\begin{array}{c}\lambda\cos\phi\\\lambda\sin\phi\end{array}\right),
\end{equation}
so that \eqref{c13} becomes
\begin{equation}
\label{c16}
    \frac{1}{M^2}\iint\,dt\,dt'\,
    \Norm_{{P}_{\mathrm{c},k}-1/2,\sigma_{\mathrm{c}}^2/M}(t/M-1/2)
    \Norm_{{P}_{\mathrm{s},k}-1/2,\sigma_s^2/M}(t'/M-1/2)\delta(AF(t,t')).
\end{equation}
Using the fact that
\begin{equation}
    \int f(x)\delta(g(x))\,dx = \frac{f(0)}{|g'(0)|},
\end{equation}
\eqref{c16} becomes
\begin{equation}
    \frac{1}{AM}\int\,dt\,
    \Norm_{{P}_{\mathrm{c},k}-1/2,\sigma_{\mathrm{c}}^2/M}(t/M-1/2)
    \Norm_{{P}_{\mathrm{s},k}-1/2,\sigma_s^2/M}((t/M-1/2)\cot\hat\phi).
\end{equation}
Inserting the right-hand side of \eqref{prob_params} and pulling the $\cot\hat\phi$ out of the argument of the second normal distribution gives
\begin{equation}
    \frac{1}{AM\tan\hat\phi}\int\,dt\,
    \Norm_{(\lambda\cos\phi)/2,\sigma_{\mathrm{c}}^2/M}(t/M-1/2)
    \Norm_{(\lambda\sin\phi\cot\hat\phi)/2,\sigma_s^2\cot^2\hat\phi/M}(t/M-1/2).
\end{equation}
We then change variables to $\tau\equiv\frac{t}{M}-\frac{1}{2}$, resulting in
\begin{equation}
    \frac{1}{A\tan\hat\phi}\int\,d\tau\,
    \Norm_{(\lambda\cos\phi)/2,\sigma_{\mathrm{c}}^2/M}(\tau)
    \Norm_{(\lambda\sin\phi\cot\hat\phi)/2,\sigma_s^2\cot^2\hat\phi/M}(\tau).
\end{equation}
Using \cref{app:norm_pdf_prod}, we get
\begin{align}
    &\frac{1}{A\tan\hat\phi}\int\,d\tau\,
    \Norm_{\mu,\sigma}(\tau)
    \Norm_{0,(\sigma_{\mathrm{c}}^2+\sigma_s^2\cot^2\hat\phi)/M}((\lambda\cos\phi)/2-(\lambda\sin\phi\cot\hat\phi)/2)\\
    =\ & \frac{1}{A\tan\hat\phi}
    \Norm_{0,(\sigma_{\mathrm{c}}^2+\sigma_s^2\cot^2\hat\phi)/M}(\lambda(\cos\phi-\sin\phi\cot\hat\phi)/2)\\
    =\ & \frac{\cos\hat\phi}{A\lambda}
    \Norm_{0,(\sigma_{\mathrm{c}}^2\sin^2\hat\phi+\sigma_s^2\cos^2\hat\phi)/(M\lambda^2)}((\cos\phi\sin\hat\phi-\sin\phi\cos\hat\phi)/2)\\
    =\ & \frac{\cos\hat\phi}{A\lambda}
    \Norm_{0,(\sigma_{\mathrm{c}}^2\sin^2\hat\phi+\sigma_s^2\cos^2\hat\phi)/(M\lambda^2)}((\sin(\phi-\hat\phi)/2))\\
    =\ &\Norm_{0,(\sigma_{\mathrm{c}}^2\sin^2\hat\phi+\sigma_s^2\cos^2\hat\phi)/(M\lambda^2)}((\sin(\phi-\hat\phi)/2)),
\end{align}
where we have put $A^{-1}=\lambda/\cos\hat\phi$ to ensure the proper normalization of the probability distribution.
Observe that since $\sigma_{\mathrm{c}}^2,\sigma_s^2\in[0,0.25]$, in the large $M$ and small $\lambda$ limit, the quantity
\begin{equation}
    \frac{\sigma_{\mathrm{c}}^2\sin^2\hat\phi+\sigma_s^2\cos^2\hat\phi}{M\lambda^2}=\frac{O(1)}{M\lambda^2}
\end{equation}
controls the variance of the normal distribution of measured angles, i.e., the number of samples is scaled by a factor of $\lambda^2$.

\section{Limitations of an Alternative, Set Formulation of RPE}
\label[appendix]{app:interval-form}

In this appendix, we consider only the case that $N_k=2^{k}$, and discuss an alternative formulation of RPE.
In standard RPE, a single angle $\hat\theta_k$ is selected at each generation.
Instead, one might imagine identifying sets $\Omega_k$ of permitted angles, selected on the basis of their proximity to $\Theta_k$.
Contrast this with the sets of angles used in criteria 1, 2, and 3 developed above, which are defined by proximity to a single selected angle $\hat\theta_k$.
In particular, these $\Omega_k$ are not necessarily a single interval, and may be computationally intensive to track.
One might therefore expect to obtain additional deductive strength by using these sets, since they require significantly more classical power to manage.
However, we show that this alternative formulation cannot tolerate errors greater than $\pi/3$ without failing to exclude infinitely many false candidate values for the angle, and therefore it provides no advantage over standard RPE.

To make this protocol precise, assume that at every generation the measurements suffer an error no greater than some fixed angle $2\pi\alpha$,
\begin{equation}
d(2\pi \tau,\Theta_k)<(2\pi\alpha) 2^{-k},\label{eqn:putative_limit}
\end{equation}
where $2\pi\tau$ is the ``true'' angle we are attempting to measure, and we are guaranteeing that the measurements are sufficiently accurate.
One may ask the question, ``Can we relax $2\pi\alpha$ above the $\pi/3$ uniform-approximation limit (the limit for standard RPE) and still obtain a valid estimate of the true angle $2\pi\tau$?''
Here, we provide a counterexample by showing that, for any $2\pi\alpha>\pi/3$ and integer $j\geq 0$, there is a sequence of measurements satisfying \cref{eqn:putative_limit} that always includes a false angle $2\pi\phi$, defined by
\begin{equation}\label{eq:phi_def}
2\pi\phi=2\pi\frac{2^{-j}}{3}.
\end{equation}
In other words, in this relaxed case we can find measurements satisfying the error bound that converge to any one of infinitely many incorrect $2\pi\phi$.

To construct the counterexample, we first clarify the exact process of this generalized set formulation of RPE, which we parameterize by some $\{\beta_k\}_{k\geq 0}$, with $\beta_k>0$.
As with standard RPE, at every generation, $N_k$ candidate values for $\theta$ are provided in $\Theta_k$.
We define the permitted subset $\Omega_k$ of the angular space as
\begin{equation}
\label{app_eq:Omega_k_def}
\Omega_k=\Omega_{k-1}\cap \left\{ \theta\ \middle|\ \exists \tilde\theta_k\in\Theta_k \ :\ \left| \tilde\theta_k-\theta \right|_{2\pi}<2\pi\beta_k \right\}
= \bigcap_{k'\leq k} \left\{ \theta\ \middle|\ \exists \tilde\theta_{k'}\in\Theta_{k'} \ :\ \left| \tilde\theta_{k'}-\theta \right|_{2\pi}<2\pi\beta_{k'} \right\},
\end{equation}
for $k\geq 0$, and letting $\Omega_{-1}=[-\pi,\pi)$.
$\beta_k$ must be at least large enough that $\Omega_k$ contains the true angle $2\pi\tau$, which can be as far as $2\pi\alpha$ from $\M(2\pi\tau,\Theta_k)$.
Therefore we also require that this generalized set formulation use
\begin{equation}
2\pi\beta_k \geq 2\pi\alpha 2^{-k}.
\end{equation}
Without loss of generality, we choose $\beta_k=\alpha 2^{-k}$ to saturate this inequality, because any larger choice of $\beta_k$ will necessarily create a superset of $\Omega_k$, and therefore also fail to exclude our pathological false angle $\phi$.

If $2\pi\phi$ is to serve as the angle in the counterexample, we need to show that $2\pi\phi$ is a member of $\Omega_k$ for every $k\geq0$.
This will follow immediately from \cref{app_eq:Omega_k_def} if we show that $2\pi\phi$ is within $2\pi\beta_k$ of $\Theta_k$ for each $k$, i.e.,
\begin{equation}\label{app_genform:counteexample_condition}
2\pi\phi \in \left\{ \theta\ \middle|\ \exists \tilde\theta_k\in\Theta_k \ :\ \left| \tilde\theta_k-\theta\right|_{2\pi}<2\pi\beta_k \right\}.
\end{equation}
Assume without loss of generality that $\tau=0$.
Let the error in the measured angle at generation $k$ be $2\pi\epsilon_k$.
Then all elements of $\Theta_k$ incur an error of $2\pi\frac{\epsilon_k}{N_k}$, so
\begin{equation}
 \Theta_k = \left\{ \frac{2\pi}{N_k}(q+\epsilon_k) ~\middle|~ q\in \mathbb{Z}_{N_k}\right\}.
\end{equation}
If we insert this into \cref{app_genform:counteexample_condition}, i.e., replace $\tilde\theta_k$ with $\frac{2\pi}{N_k}(q+\epsilon_k)$, then \cref{app_genform:counteexample_condition} is equivalent to
\begin{equation}
\exists q\in \mathbb{Z}_{N_k} \ :\ 2\pi\left| \frac{1}{N_k}\left(q+\epsilon_k\right)-\phi\right|_{2\pi} < 2\pi\alpha2^{-k},
\end{equation}
i.e.,
\begin{equation}\label{app_genform:counteexample_condition2}
\exists q\in \mathbb{Z}_{N_k} \ :\ \left| q+\epsilon_k-\phi N_k \right|_{2\pi} < \alpha.
\end{equation}
If we define $F(x)=x- \lfloor x\rfloor$ to be the fractional part of the real number $x$, then \cref{app_genform:counteexample_condition2} becomes
\begin{equation}
\label{app_genform:counteexample_condition3}
F(\epsilon_k-\phi 2^k) < \alpha\quad\text{or}\quad F(\epsilon_k-\phi 2^k) > 1-\alpha.
\end{equation}
Recall that $|\epsilon_k|\le\alpha$, since $\alpha$ is defined to be the maximum allowed error.
Therefore, for any $\phi$ satisfying
\begin{equation}
\label{app_genform:counteexample_condition4}
F(\phi 2^k) < 2\alpha\quad\text{or}\quad F(\phi 2^k) > 1-2\alpha,
\end{equation}
there exists $\epsilon_k$ such that $\phi$ satisfies \cref{app_genform:counteexample_condition3}, i.e., $\phi$ is a possible angle in the counterexample.
Equivalently,
\begin{equation}\label{eqn:frac_tau}
\left| F(\phi 2^k) - \frac{1}{2}\right| > \frac{1}{2} - 2\alpha.
\end{equation}

Now, suppose $\alpha>\frac{1}{6}$.
Then $\frac{1}{6}>\frac{1}{2}-2\alpha$.
It therefore suffices to satisfy
\begin{equation}
\label{app_genform:c12}
\left| F(\phi 2^k) - \frac{1}{2}\right| > \frac{1}{6}
\end{equation}
in order to satisfy \cref{eqn:frac_tau}.
If $\phi=\frac{2^{-j}}{3}$ with $j\leq k$, $F(\phi 2^k)$ is either $\frac{1}{3}$ or $\frac{2}{3}$, both of which satisfy \cref{app_genform:c12}.
For $j<k$, $\phi 2^k<1/3$, and \cref{app_genform:c12} is again also satisfied.
This proves the claim.

\end{document}